\documentclass[11pt,amssymb,amsfont,a4paper]{article}
\usepackage{latexsym,amssymb,amsmath,amsthm,color}
\usepackage{geometry}
\usepackage{url}
\usepackage[utf8]{inputenc}
\usepackage{authblk}

\theoremstyle{definition}
\newtheorem{definition}{Definition}
\newtheorem{example}{Example}

\theoremstyle{plain}
\newtheorem{theorem}{Theorem}
\newtheorem{proposition}{Proposition}
\newtheorem{lemma}{Lemma}
\newtheorem{remark}{Remark}
\newtheorem{corollary}{Corollary}

\title{Rank error-correcting pairs}
\author[1]{Umberto Mart{\'i}nez-Pe\~{n}as \thanks{umberto@math.aau.dk}}
\author[2]{Ruud Pellikaan \thanks{g.r.pellikaan@tue.nl}}
\affil[1]{Department of Mathematical Sciences, Aalborg University, Denmark}
\affil[2]{Department of Mathematics and Computing Science, Eindhoven University of Technology, The Netherlands}

\begin{document}

\maketitle

\begin{abstract}
Error-correcting pairs were introduced independently by Pellikaan and K{\"o}tter as a general method of decoding linear codes with respect to the Hamming metric using coordinatewise products of vectors, and are used for many well-known families of codes. In this paper, we define new types of vector products, extending the coordinatewise product, some of which preserve symbolic products of linearized polynomials after evaluation and some of which coincide with usual products of matrices. Then we define rank error-correcting pairs for codes that are linear over the extension field and for codes that are linear over the base field, and relate both types. Bounds on the minimum rank distance of codes and MRD conditions are given. Finally we show that some well-known families of rank-metric codes admit rank error-correcting pairs, and show that the given algorithm generalizes the classical algorithm using error-correcting pairs for the Hamming metric.


\end{abstract}

\section{Introduction}

Error-correcting pairs were introduced independently by Pellikaan in \cite{on-ECP, pellikaan-error-location} and by K{\"o}tter in \cite{unified}. These are pairs of linear codes satisfying some conditions with respect to the coordinatewise product and a given linear code, for which they define an error-correcting algorithm with respect to the Hamming metric in polynomial time. 

Linear codes with an error-correcting pair include many well-known families, such as (generalized) Reed-Solomon codes, many cyclic codes (such as BCH codes), Goppa codes and algebraic geometry codes (see \cite{error-locating, pellikaan-error-location, on-the-existence}).

Error-correcting codes with respect to the rank metric \cite{gabidulin} have recently gained considerable attention due to their applications in network coding \cite{on-metrics}. In the rank metric, maximum rank distance (MRD) Gabidulin codes, as defined in \cite{gabidulin, new-construction}, have been widely used, and decoding algorithms using linearized polynomials are given in \cite{gabidulin, new-construction, loidreau-gabidulin}. A related construction, the so-called $ q $-cyclic or skew cyclic codes, were introduced by Gabidulin in \cite{gabidulin} for square matrices and generalized independently by himself in \cite{qcyclic} and by Ulmer et al. in \cite{skewcyclic1}.

However, more general methods of decoding with respect to the rank metric are lacking, specially for codes that are linear over the base field instead of the extension field.

The contributions of this paper are organized as follows. In Section 3, we introduce some families of vector products that coincide with usual products of matrices for some sizes. One of these products preserves symbolic products of linearized polynomials after evaluation and is the unique product with this property for some particular sizes. In Section 4, we introduce the concept of rank error-correcting pair and give efficient decoding algorithms based on them. Subsection 4.1 treats linear codes over the extension field, and Subsection 4.2 treats linear codes over the base field. In Section 5, we prove that the latter type of rank error-correcting pairs generalize the former type. In Section 6, we derive bounds on the minimum rank distance and give MRD conditions based on rank error-correcting pairs. Finally, in Section 7, we study some families of codes that admit rank error-correcting pairs, showing that the given algorithm generalizes the classical algorithm using error-correcting pairs for the Hamming metric.

\section{Preliminaries} \label{preliminaries}

Fix a prime power $ q $ and positive integers $ m $ and $ n $. Fix a basis $ \alpha_1, \alpha_2, \ldots, \alpha_m $ of $ \mathbb{F}_{q^m} $ over $ \mathbb{F}_q $. We will use the following classical matrix representation of vectors in $ \mathbb{F}_{q^m}^n $. Let $ \mathbf{c} \in \mathbb{F}_{q^m}^n $, there exist unique $ \mathbf{c}_i \in \mathbb{F}_q^n $, for $ i = 1,2, \ldots, m $, such that $ \mathbf{c} = \sum_{i=1}^m \alpha_i \mathbf{c}_i $. Let $ \mathbf{c}_i = (c_{i,1}, c_{i,2}, \ldots, c_{i,n}) $ or, equivalently, $ \mathbf{c} = (c_1, c_2, \ldots, c_n) $ and $ c_j = \sum_{i=1}^m \alpha_i c_{i,j} $. Then we define the $ m \times n $ matrix, with coefficients in $ \mathbb{F}_q $,
\begin{equation} \label{matrix representation}
M(\mathbf{c}) = (c_{i,j})_{1 \leq i \leq m, 1 \leq j \leq n}.
\end{equation}
The map $ M : \mathbb{F}_{q^m}^n \longrightarrow \mathbb{F}_q^{m \times n} $ is an $ \mathbb{F}_q $-linear vector space isomorphism, where $ \mathbb{F}_q^{m \times n} $ represents the space of $ m \times n $ matrices over $ \mathbb{F}_q $. Unless it is necessary, we will not write subscripts for $ M $ regarding the values $ m $, $ n $, or the basis $ \alpha_1, \alpha_2, \ldots, \alpha_m $ (which of course change the map $ M $).

By definition \cite{gabidulin}, the rank weight of $ \mathbf{c} $ is $ {\rm wt_R}(\mathbf{c}) = {\rm Rk}(M(\mathbf{c})) $, the rank of the matrix $ M(\mathbf{c}) $, for every $ \mathbf{c} \in \mathbb{F}_{q^m}^n $. We also define the rank support of $ \mathbf{c} $ as the row space of the matrix $ M(\mathbf{c}) $, that is, $ {\rm RSupp}(\mathbf{c}) = {\rm Row}(M(\mathbf{c})) \subseteq \mathbb{F}_q^n $. We may identify any code $ \mathcal{C} \subseteq \mathbb{F}_{q^m}^n $ with $ M(\mathcal{C}) \subseteq \mathbb{F}_q^{m \times n} $ and write $ d_R(\mathcal{C}) = d_R(M(\mathcal{C})) $ for their minimum rank distance \cite{gabidulin}.

We also define the map $ D : \mathbb{F}_q^n \longrightarrow \mathbb{F}_q^{n \times n} $ as follows. For every vector $ \mathbf{c} \in \mathbb{F}_q^n $, define the matrix
\begin{equation} \label{diagonal}
D(\mathbf{c}) = {\rm diag}(\mathbf{c}) = (c_i \delta_{i,j})_{1 \leq i \leq n, 1 \leq j \leq n},
\end{equation}
that is, the diagonal $ n \times n $ matrix with coefficients in $ \mathbb{F}_q $ whose diagonal vector is $ \mathbf{c} $. The map $ D $ is $ \mathbb{F}_q $-linear and one to one. Moreover, the Hamming weight of a vector $ \mathbf{c} \in \mathbb{F}_q^n $ is $ {\rm wt_H}(\mathbf{c}) = {\rm Rk}(D(\mathbf{c})) $.

This gives a way to represent error-correcting codes $ \mathcal{C} \subseteq \mathbb{F}_q^n $ in the Hamming metric as error-correcting codes $ D(\mathcal{C}) \subseteq \mathbb{F}_q^{n \times n} $ in the rank metric, where the Hamming weight distribution of $ \mathcal{C} $ corresponds bijectively to the rank weight distribution of $ D(\mathcal{C}) $. In particular, the minimum Hamming distance of $ \mathcal{C} $ satisfies $ d_H(\mathcal{C}) = d_R(D(\mathcal{C})) $. 

On the other hand, an $ \mathbb{F}_q $-linear map $ \phi : \mathbb{F}_q^n \longrightarrow \mathbb{F}_q^n $ (respectively, an $ \mathbb{F}_{q^m} $-linear map $ \phi : \mathbb{F}_{q^m}^n \longrightarrow \mathbb{F}_{q^m}^n $) is a Hamming-metric equivalence (respectively, rank-metric equivalence) if $ {\rm wt_H}(\phi(\mathbf{c})) = {\rm wt_H}(\mathbf{c}) $ (respectively, $ {\rm wt_R}(\phi(\mathbf{c})) = {\rm wt_R}(\mathbf{c}) $), for all $ \mathbf{c} \in \mathbb{F}_q^n $ (respectively, $ \mathbf{c} \in \mathbb{F}_{q^m}^n $). 

Let $ \phi : \mathbb{F}_q^n \longrightarrow \mathbb{F}_q^n $ be a Hamming-metric equivalence. It is well-known that $ \phi $ is a monomial map, that is, there exist $ a_1,a_2, \ldots, a_n \in \mathbb{F}_q^* $ and a permutation $ \sigma $ with $ \phi(\mathbf{c}) = (a_1c_{\sigma(1)}, a_2c_{\sigma(2)}, \ldots, a_nc_{\sigma(n)}) $, for all $ \mathbf{c} \in \mathbb{F}_q^n $. Then $ \phi^\prime : \mathbb{F}_{q^n}^n \longrightarrow \mathbb{F}_{q^n}^n $, defined by the same formula, is an $ \mathbb{F}_{q^n} $-linear rank-metric equivalence map (see \cite{berger, similarities}), and $ M^{-1}(D(\phi(\mathbf{c}))) = \phi^\prime (M^{-1}(D(\mathbf{c}))) $, for all $ \mathbf{c} \in \mathbb{F}_q^n $. This implies that Hamming-metric equivalent codes $ \mathcal{C}_1 $ and $ \mathcal{C}_2 $ in $ \mathbb{F}_q^n $ correspond to rank-metric equivalent codes $ D(\mathcal{C}_1) $ and $ D(\mathcal{C}_2) $ in $ \mathbb{F}_q^{n \times n} $.

We will also use the following notation. Given a subset $ \mathcal{A} \subseteq \mathbb{F}_{q^m}^n $, we denote by $ \langle \mathcal{A} \rangle_{\mathbb{F}_q} $ and $ \langle \mathcal{A} \rangle_{\mathbb{F}_{q^m}} $ the $ \mathbb{F}_q $-linear and $ \mathbb{F}_{q^m} $-linear vector spaces generated by $ \mathcal{A} $, respectively. For an $ \mathbb{F}_{q^m} $-linear (respectively $ \mathbb{F}_q $-linear) code $ \mathcal{C} \subseteq \mathbb{F}_{q^m}^n $ (respectively $ \mathcal{C} \subseteq \mathbb{F}_q^n $), we denote its dimension over $ \mathbb{F}_{q^m} $ (respectively over $ \mathbb{F}_q $) by $ \dim(\mathcal{C}) $. If $ \mathcal{C} \subseteq \mathbb{F}_{q^m}^n $ or $ \mathcal{C} \subseteq \mathbb{F}_q^{m \times n} $ is $ \mathbb{F}_q $-linear, we denote its dimension over $ \mathbb{F}_q $ by $ \dim_{\mathbb{F}_q}(\mathcal{C}) $. \\

We conclude by defining error-correcting pairs (ECPs) for the Hamming metric, introduced independently by Pellikaan in \cite{on-ECP, pellikaan-error-location} and by K{\"o}tter in \cite{unified}. Define the coordinatewise product $ \ast $ of vectors in $ \mathbb{F}_q^n $ by
$$ \mathbf{a} \ast \mathbf{b} = (a_1b_1, a_2b_2, \ldots, a_nb_n), $$
for all $ \mathbf{a}, \mathbf{b} \in \mathbb{F}_q^n $. For two linear subspaces $ \mathcal{A}, \mathcal{B} \subseteq \mathbb{F}_q^n $, we define the linear subspace $ \mathcal{A} \ast \mathcal{B} = \langle \{ \mathbf{a} \ast \mathbf{b} \mid \mathbf{a} \in \mathcal{A}, \mathbf{b} \in \mathcal{B} \} \rangle \subseteq \mathbb{F}_q^n $. 

\begin{definition}
Let $ \mathcal{A}, \mathcal{B}, \mathcal{C} \subseteq \mathbb{F}_q^n $ be linear codes and $ t $ a positive integer. The pair $ (\mathcal{A}, \mathcal{B}) $ is called a $ t  $-error-correcting pair ($ t $-ECP) for $ \mathcal{C} $ if the following properties hold:
\begin{enumerate}
\item
$ \mathcal{A} \ast \mathcal{B} \subseteq \mathcal{C}^\perp $.
\item
$ \dim(\mathcal{A}) > t $.
\item
$ d_H(\mathcal{B}^\perp) > t $.
\item
$ d_H(\mathcal{A}) + d_H(\mathcal{C}) > n $.
\end{enumerate}
\end{definition}

In \cite{on-ECP, pellikaan-error-location} it is shown that, if $ \mathcal{C} $ has a $ t $-ECP, then it has a decoding algorithm with complexity $ O(n^3) $ that can correct up to $ t $ errors in the Hamming metric (and therefore, $ d_H(\mathcal{C}) \geq 2t + 1 $). This algorithm is analogous to the ones that we will describe in Subsections \ref{using extension} and \ref{using base}. Actually, as we will see in Subsection \ref{subsection Hamming}, the algorithm presented in Subsection \ref{using base} extends the classical algorithm for Hamming-metric codes.

\section{Vector products for the rank metric} \label{vector products}

In this section, we define and give the basic properties of a family of products of vectors in $ \mathbb{F}_{q^m}^n $, which will play the same role as the coordinatewise product $ \ast $ for vectors in $ \mathbb{F}_q^n $.

\begin{definition} \label{definition product}
We first define the product $ \star : \mathbb{F}_{q^m}^m \times \mathbb{F}_{q^m}^n \longrightarrow \mathbb{F}_{q^m}^n $ in the following way. For every $ \mathbf{c} \in \mathbb{F}_{q^m}^m $ and every $ \mathbf{d} \in \mathbb{F}_{q^m}^n $, we define
$$ \mathbf{c} \star \mathbf{d} = \sum_{i=1}^m c_i \mathbf{d}_i, $$
where $ \mathbf{d} = \sum_{i=1}^m \alpha_i \mathbf{d}_i $ and $ \mathbf{d}_i \in \mathbb{F}_q^n $, for all $ i $, and $ \mathbf{c} = (c_1, c_2, \ldots, c_m) $. Note that the second argument of $ \star $ and its codomain are the same, whereas its first argument is different if $ m \neq n $.

On the other hand, given a map $ \varphi : \mathbb{F}_{q^m}^n \longrightarrow \mathbb{F}_{q^m}^m $, we define the product $ \star_\varphi : \mathbb{F}_{q^m}^n \times \mathbb{F}_{q^m}^n \longrightarrow \mathbb{F}_{q^m}^n $ in the following way. For every $ \mathbf{c}, \mathbf{d} \in \mathbb{F}_{q^m}^n $, we define
$$ \mathbf{c} \star_\varphi \mathbf{d} = \varphi(\mathbf{c}) \star \mathbf{d} = \sum_{i=1}^m \varphi(\mathbf{c})_i \mathbf{d}_i, $$
where $ \mathbf{d} = \sum_{i=1}^m \alpha_i \mathbf{d}_i $ and $ \mathbf{d}_i \in \mathbb{F}_q^n $, for all $ i $, and $ \varphi(\mathbf{c}) = (\varphi(\mathbf{c})_1, \varphi(\mathbf{c})_2, \ldots, \varphi(\mathbf{c})_m) $.
\end{definition}

\begin{remark} \label{bilinearity}
The following basic properties of the previous products hold:
\begin{enumerate}
\item
The product $ \star $ depends on the choice of the basis $ \alpha_1, \alpha_2, \ldots, \alpha_m $ of $ \mathbb{F}_{q^m} $ over $ \mathbb{F}_q^n $, whereas the coordinatewise product $ \ast $ does not.
\item
The product $ \star $ is $ \mathbb{F}_{q^m} $-linear in the first component and $ \mathbb{F}_q $-linear in the second component.
\item
If $ \varphi $ is $ \mathbb{F}_q $-linear, then the product $ \star_\varphi $ is $ \mathbb{F}_q $-bilinear.
\item
On the other hand, if $ \varphi $ is $ \mathbb{F}_{q^m} $-linear, then the product $ \star_\varphi $ is $ \mathbb{F}_{q^m} $-linear in the first component and $ \mathbb{F}_q $-linear in the second component.
\end{enumerate}
\end{remark}

It is of interest to see if two maps give the same product:

\begin{lemma} \label{unique map}
Given maps $ \varphi, \psi : \mathbb{F}_{q^m}^n \longrightarrow \mathbb{F}_{q^m}^m $, it holds that $ \star_\varphi = \star_\psi $ if, and only if, $ \varphi = \psi $.
\end{lemma}
\begin{proof}
Fix $ i $ and take $ \mathbf{d} \in \mathbb{F}_{q^m}^n $ such that $ \mathbf{d}_i = \mathbf{e}_1 $, the first vector in the canonical basis of $ \mathbb{F}_q^n $ and $ \mathbf{d}_j = \mathbf{0} $, for $ j \neq i $. Since $ \mathbf{c} \star_\varphi \mathbf{d} = \mathbf{c} \star_\psi \mathbf{d} $, it follows that $ \varphi(\mathbf{c})_i = \psi (\mathbf{c})_i $. This is valid for an arbitrary $ i $, hence $ \varphi (\mathbf{c}) = \psi (\mathbf{c}) $, for any $ \mathbf{c} \in \mathbb{F}_{q^m}^n $, which implies that $ \varphi = \psi $. The reversed implication is trivial.
\end{proof}

One of the most important properties of the coordinatewise product $ \ast $ is that it preserves multiplications of polynomials after evaluation.

We now define a natural product that will preserve symbolic multiplications of linearized polynomials after evaluation. Recall that a $ q $-linearized polynomial over $ \mathbb{F}_{q^m} $ is a polynomial of the form
$$ F = a_0 x + a_1 x^{[1]} + \cdots + a_d x^{[d]}, $$
where $ a_0,a_1, \ldots, a_d \in \mathbb{F}_{q^m} $ and $ [i] = q^i $, for all $ i $. These polynomials induce $ \mathbb{F}_q $-linear maps in any extension field of $ \mathbb{F}_{q^m} $. We start by the following interpolation lemma, where we denote by $ \mathcal{L}_q\mathbb{F}_{q^m}[x] $ the set of $ q $-linearized polynomials over $ \mathbb{F}_{q^m} $. 

\begin{lemma} \label{interpolation}
If $ n \leq m $, and $ \mathbf{c} \in \mathbb{F}_{q^m}^n $, there exists a unique $ q $-polynomial $ F \in \mathcal{L}_q\mathbb{F}_{q^m}[x] $ of degree less than $ q^n = [n] $ such that $ F(\alpha_i) = c_i $, for all $ i =1,2, \ldots, n $.
\end{lemma}
\begin{proof}
Consider the evaluation map $ {\rm ev}_{\boldsymbol\alpha} : \mathcal{L}_q\mathbb{F}_{q^m}[x] \longrightarrow \mathbb{F}_{q^m}^n $, defined by $ {\rm ev}_{\boldsymbol\alpha}(F) = (F(\alpha_1), F(\alpha_2), \ldots, F(\alpha_n)) $.

Since it is $ \mathbb{F}_{q^m} $-linear and the $ \mathbb{F}_{q^m} $-linear space of $ q $-linearized polynomials of degree less than $ [n] $ has dimension $ n $, it is enough to prove that, if $ F(\alpha_i) = 0 $, for $ i=1,2, \ldots, n $, then $ F = 0 $.

By the linearity of $ F $, we have that $ F(\sum_i \lambda_i\alpha_i) = \sum_i \lambda_i F(\alpha_i) = 0 $, for every $ \lambda_1, \lambda_2, \ldots, \lambda_n \in \mathbb{F}_q $. Therefore, $ F $ has $ q^n $ different roots and degree less than $ q^n $, hence $ F = 0 $, and we are done.
\end{proof}

From now on, if $ n \leq m $, we denote by $ F_\mathbf{c} $ the $ q $-linearized polynomial of degree less than $ [n] $ corresponding to $ \mathbf{c} \in \mathbb{F}_{q^m}^n $. In the following remark we show how to perform interpolation using symbolic multiplications of linearized polynomials. Recall that the symbolic multiplication of two linearized polynomials $ F, G \in \mathcal{L}_q\mathbb{F}_{q^m}[x] $ is defined as their composition $ F \circ G $, which lies in $ \mathcal{L}_q\mathbb{F}_{q^m}[x] $.

\begin{remark}
Interpolation as presented in the previous lemma can be performed as follows. First, we see that the map $ \mathbf{c} \in \mathbb{F}_{q^m}^n \mapsto F_\mathbf{c} $ is $ \mathbb{F}_{q^m} $-linear. Therefore, 
$$ F_\mathbf{c} = \sum_{i=1}^n c_i F_{\mathbf{e}_i}, $$
where $ \mathbf{e}_i = (0, \ldots, 0, 1, 0, \ldots, 0) $ is the $ i $-th vector in the canonical basis of $ \mathbb{F}_{q^m}^n $ over $ \mathbb{F}_{q^m} $, for $ i = 1,2, \ldots, n $. On the other hand, it holds that
$$ F_{\mathbf{e}_i} = \frac{G_i}{G_i(\alpha_i)}, \quad \textrm{where} \quad G_i = \prod_{j \neq i} \prod_{\beta \in \langle \alpha_j \rangle} (x-\beta), $$
and where $ 1 \leq j \leq n $. The polynomial $ G_i / G_i(\alpha_i) $ in this expression is well-defined since $ \alpha_i $ does not belong to the $ \mathbb{F}_q $-linear vector space generated by the elements $ \alpha_j $, for $ j \neq i $, and the expression in the numerator is a $ q $-linearized polynomial by \cite[Theorem 3.52]{lidl} and has degree less than $ q^n $. However, the complexity of constructing $ G_i $ in this way is of $ O(q^{n-1}) $ conventional multiplications. The following expression shows how to compute $ G_i $ with $ O(n-1) $ symbolic multiplications:
$$ G_i = L_{i,n} \circ L_{i,n-1} \circ \cdots \circ \widehat{L}_{i,i} \circ \cdots \circ L_{i,2} \circ L_{i,1}, $$
where $ L_{i,1} = x^{[1]} - ( \alpha_1^{[1]} / \alpha_1 )x $ and, for $ j = 2,3, \ldots, n $, 
$$ L_{i,j} = x^{[1]} - (\widetilde{L}_{i,j-1} (\alpha_j)^{[1]} / \widetilde{L}_{i,j-1} (\alpha_j) )x $$ 
and $ \widetilde{L}_{i,j} = L_{i,j-1} \circ \cdots \circ \widehat{L}_{i,i} \circ \cdots \circ L_{i,2} \circ L_{i,1} $. The notation $ \widehat{L}_{i,i} $ means that the polynomial $ L_{i,i} $ is omitted.
\end{remark}

On the other hand, for $ \mathbf{c} \in \mathbb{F}_{q^m}^n $ and $ n \leq m $, we define the vector $ \varphi_n(\mathbf{c}) \in \mathbb{F}_{q^m}^m $ as $ \varphi_n(\mathbf{c})_i = F_\mathbf{c}(\alpha_i) $, for $ i = 1,2, \ldots, m $. If $ n \geq m $, we define $ \varphi_n(\mathbf{c}) = (c_1,c_2, \ldots, c_m) $. 

Note that if $ n=m $, both definitions lead to $ \varphi_n(\mathbf{c}) = \mathbf{c} $. Also note that $ \varphi_n $ depends on the basis $ \alpha_1, \alpha_2, \ldots, \alpha_m $ for $ n < m $, while it does not for $ n \geq m $.

\begin{lemma} \label{phi is linear}
For any values of $ m $ and $ n $, the map $ \varphi_n : \mathbb{F}_{q^m}^n \longrightarrow \mathbb{F}_{q^m}^m $ is $ \mathbb{F}_{q^m} $-linear.
\end{lemma}
\begin{proof}
For $ n \geq m $, it is clear. For $ n \leq m $, it is enough to note that $ F_{\gamma \mathbf{c} + \delta \mathbf{d}} = \gamma F_\mathbf{c} + \delta F_\mathbf{d} $ as in the remark above, for all $ \gamma, \delta \in \mathbb{F}_{q^m} $ and all $ \mathbf{c}, \mathbf{d} \in \mathbb{F}_{q^m}^n $.
\end{proof}

We will use the notation $ \star = \star_{\varphi_n} : \mathbb{F}_{q^m}^n \times \mathbb{F}_{q^m}^n \longrightarrow \mathbb{F}_{q^m}^n $. When $ m=n $, this coincides with the product $ \star $ in Definition \ref{definition product}, whereas if $ m \neq n $, then there is no confusion with the product $ \star $ in Definition \ref{definition product}, since the first argument is different. Hence the meaning of $ \star $ is clear from the context.

The interesting property of the product $ \star $ is that it preserves symbolic multiplications of linearized polynomials, as we will see now, and in the case $ n \leq m $, it is the unique product with this property. 

For a vector $ \mathbf{b} \in \mathbb{F}_{q^m}^n $, we will define the evaluation map
$$ {\rm ev}_\mathbf{b} : \mathcal{L}_q\mathbb{F}_{q^m}[x] \longrightarrow \mathbb{F}_{q^m}^n $$
by $ {\rm ev}_\mathbf{b}(F) = (F(b_1), F(b_2), \ldots, F(b_n)) $.

From now on, we denote $ \boldsymbol\alpha_n = (\alpha_1, \alpha_2, \ldots, \alpha_n) $ if $ n \leq m $, and we complete the vector with other elements if $ n > m $, $ \boldsymbol\alpha_n = (\alpha_1, \alpha_2, \ldots, \alpha_m, \gamma_1, \gamma_2, \ldots, \gamma_n) $. We will also denote $ \boldsymbol\alpha = (\alpha_1, \alpha_2, \ldots, \alpha_m) $. Observe that $ \varphi_n(\boldsymbol\alpha_n) = \boldsymbol\alpha $ in all cases, and moreover, $ \varphi_n(\boldsymbol\alpha_n^{[j]}) = \boldsymbol\alpha^{[j]} $, if $ j < n $.

\begin{proposition} \label{evaluation properties}
The following properties hold:
\begin{enumerate}
\item
$ \boldsymbol\alpha^{[j]} \star \mathbf{c} = \mathbf{c}^{[j]} $, for all $ \mathbf{c} \in \mathbb{F}_{q^m}^n $ and all $ j $. In particular, 
$$ {\rm ev}_\mathbf{b}(F \circ G) = {\rm ev}_{\boldsymbol\alpha}(F) \star {\rm ev}_\mathbf{b}(G), $$ 
for all $ \mathbf{b} \in \mathbb{F}_{q^m}^n $ and all $ F,G \in \mathcal{L}_q\mathbb{F}_{q^m}[x] $.
\item
$ \boldsymbol\alpha_n^{[j]} \star \mathbf{c} = \mathbf{c}^{[j]} $, for all $ \mathbf{c} \in \mathbb{F}_{q^m}^n $ and all $ j < n $. In particular, 
$$ {\rm ev}_\mathbf{b}(F \circ G) = {\rm ev}_{\boldsymbol\alpha_n}(F) \star {\rm ev}_\mathbf{b}(G), $$ 
for all $ \mathbf{b} \in \mathbb{F}_{q^m}^n $ and all $ F,G \in \mathcal{L}_q\mathbb{F}_{q^m}[x] $, where $ F $ has degree less than $ [n] $.
\item
If $ n \leq m $, then $ \star $ is associative, that is, $ \mathbf{a} \star (\mathbf{b} \star \mathbf{c}) = (\mathbf{a} \star \mathbf{b}) \star \mathbf{c} $, for all $ \mathbf{a}, \mathbf{b}, \mathbf{c} \in \mathbb{F}_{q^m}^n $.
\end{enumerate}
Moreover, if $ n \leq m $, and if $ \odot $ is another product that satisfies item 2 for $ \mathbf{b} = \boldsymbol\alpha_n $ (or item 1 for $ \mathbf{b} = \boldsymbol\alpha $), then $ \odot = \star $. In particular, by Lemma \ref{unique map}, if $ \star_\varphi $ satisfies this property, then $ \varphi = \varphi_n $.
\end{proposition}
\begin{proof}
\begin{enumerate}
\item
The first part follows from the following chain of equalities:
$$ \boldsymbol\alpha^{[j]} \star \mathbf{c} = \sum_{i=1}^m \alpha_i^{[j]} \mathbf{c}_i = \left( \sum_{i=1}^m \alpha_i \mathbf{c}_i \right)^{[j]} = \mathbf{c}^{[j]}. $$
The second part follows from the first part, since $ \boldsymbol\alpha^{[j]} = {\rm ev}_{\boldsymbol\alpha}(x^{[j]}) $ and $ \star $ is $ \mathbb{F}_{q^m} $-linear in the first component, by Remark \ref{bilinearity} and Lemma \ref{phi is linear}.
\item
It follows from item 1, since $ \varphi_n(\boldsymbol\alpha_n^{[j]}) = \boldsymbol\alpha^{[j]} $, if $ j < n $.
\item
It follows from item 2, since $ {\rm ev}_{\boldsymbol\alpha_n} $ is surjective (by Lemma \ref{interpolation}) and symbolic multiplication of linearized polynomials is associative.
\end{enumerate}
If $ n \leq m $, the last part of the proposition follows from the fact that $ {\rm ev}_{\boldsymbol\alpha_n} $ (or $ {\rm ev}_{\boldsymbol\alpha} $) is surjective, which follows from Lemma \ref{interpolation}. \\
\end{proof}

We will now give a matrix representation of the products $ \star_\varphi $, and show that the product $ \star $ actually extends the product $ \ast $. For that purpose, we define the ``extension'' map $ E : \mathbb{F}_q^n \longrightarrow \mathbb{F}_{q^n}^n $ by $ E = M^{-1} \circ D $, which is $ \mathbb{F}_q $-linear and one to one. In other words,
\begin{equation}
E(\mathbf{c}) = (\alpha_1 c_1, \alpha_2 c_2, \ldots, \alpha_n c_n),
\label{extension map}
\end{equation}
for all $ \mathbf{c} \in \mathbb{F}_q^n $, which satisfies that $ {\rm wt_R}(E(\mathbf{c})) = {\rm wt_H}(\mathbf{c}) $. We gather in the next proposition the relations between the products $ \star_\varphi $ and $ \ast $, and the maps $ M, D $ and $ E $. The proof is straightforward.

\begin{proposition} \label{connection vector products}
For all values of $ m $ and $ n $, all maps $ \varphi $ and all $ \mathbf{c}^\prime \in \mathbb{F}_{q^m}^m $ and $ \mathbf{c}, \mathbf{d} \in \mathbb{F}_{q^m}^n $, we have that
$$ M(\mathbf{c}^\prime \star \mathbf{d}) = M(\mathbf{c}^\prime) M(\mathbf{d}) \quad \textrm{and} \quad M(\mathbf{c} \star_\varphi \mathbf{d}) = M(\varphi(\mathbf{c})) M(\mathbf{d}). $$
On the other hand, if $ m = n $ and $ \mathbf{a}, \mathbf{b} \in \mathbb{F}_q^n $, then
$$ D(\mathbf{a} \ast \mathbf{b}) = D(\mathbf{a}) D(\mathbf{b}) \quad \textrm{and} \quad E(\mathbf{a} \ast \mathbf{b}) = E(\mathbf{a}) \star E(\mathbf{b}). $$
\end{proposition}

Hence, the product $ \star : \mathbb{F}_{q^m}^m \times \mathbb{F}_{q^m}^n \longrightarrow \mathbb{F}_{q^m}^n $ is just the usual product of $ m \times m $ matrices with $ m \times n $ matrices over $ \mathbb{F}_q $, whereas the products $ \star_\varphi $ are also products of matrices after expanding the $ m \times n $ matrix in the first argument to an $ m \times m $ matrix over $ \mathbb{F}_q $.

\section{Rank error-correcting pairs}

We will define in this section error-correcting pairs (ECPs) for the rank metric, using the products $ \star $ and $ \star_\varphi $. However, which inner product to use for defining orthogonality and duality in $ \mathbb{F}_{q^m}^n $, or in $ \mathbb{F}_q^{m \times n} $, is not clear. First of all, we will always use the standard ($ \mathbb{F}_q $-bilinear) inner product $ \cdot $ in $ \mathbb{F}_q^n $. On the other hand, we will first present ECPs in $ \mathbb{F}_{q^m}^n $ that use the ($ \mathbb{F}_{q^m} $-bilinear) ``extension'' inner product,
\begin{equation}
\mathbf{c} \cdot \mathbf{d} = c_1d_1 + c_2d_2 + \cdots + c_nd_n \in \mathbb{F}_{q^m},
\label{extension inner}
\end{equation}
for all $ \mathbf{c} = (c_1,c_2, \ldots, c_n), \mathbf{d} = (d_1,d_2, \ldots, d_n) \in \mathbb{F}_{q^m}^n $, and afterwards we will use the ($ \mathbb{F}_q $-bilinear) ``base'' (or ``trace'') inner product in $ \mathbb{F}_q^{m \times n} $,
\begin{equation}
\langle C, D \rangle = \mathbf{c}_1 \cdot \mathbf{d}_1 + \mathbf{c}_2 \cdot \mathbf{d}_2 + \cdots + \mathbf{c}_m \cdot \mathbf{d}_m = {\rm Tr}(C D^T) = \sum_{i,j} c_{i,j} d_{i,j} \in \mathbb{F}_q,
\label{base inner}
\end{equation}
for $ C,D \in \mathbb{F}_q^{m \times n} $,where $ \mathbf{c}_i, \mathbf{d}_i \in \mathbb{F}_{q}^n $, for $ i = 1,2, \ldots, m $, are the rows of $ C $ and $ D $, respectively, and $ c_{i,j}, d_{i,j} \in \mathbb{F}_q $ are the entries of $ C $ and $ D $, respectively. $ {\rm Tr} $ denotes the usual trace of a square matrix.

Whereas the product $ \cdot $ is the standard $ \mathbb{F}_{q^m} $-bilinear product in $ \mathbb{F}_{q^m}^n $, the product $ \langle , \rangle $ corresponds to the standard $ \mathbb{F}_q $-bilinear product in $ \mathbb{F}_q^{mn} \cong \mathbb{F}_q^{m \times n} $. A duality theory for the product $ \langle , \rangle $ and $ \mathbb{F}_q $-linear rank-metric codes is developed originally in \cite{delsartebilinear} and further in \cite{ravagnani}, where it is also shown that duals of $ \mathbb{F}_{q^m} $-linear codes with respect to the ``extension'' inner product are equivalent to duals with respect to the ``base'' inner product (see \cite[Theorem 21]{ravagnani}). We will come back to this in Section \ref{connections}, where we will relate both kinds of error-correcting pairs.

Now we will give some relations between the product $ \star $ and the previous inner products that we will use later. If $ \mathbf{c}, \mathbf{d} \in \mathbb{F}_{q^m}^n $, $ \mathbf{d} = \sum_i \alpha_i \mathbf{d}_i $ and $ \mathbf{d}_i \in \mathbb{F}_{q}^n $, then we define
\begin{equation}
\mathbf{c}(\mathbf{d}) = (\mathbf{c} \cdot \mathbf{d}_1, \mathbf{c} \cdot \mathbf{d}_2, \ldots, \mathbf{c} \cdot \mathbf{d}_m) \in \mathbb{F}_{q^m}^m.
\label{matrix product transposed}
\end{equation}
On the other hand, given $ \mathbf{b} \in \mathbb{F}_{q^m}^m $, we define $ \mathbf{b}^T $ as the unique vector in $ \mathbb{F}_{q^m}^m $ such that $ M(\mathbf{b}^T) = M(\mathbf{b})^T $. 

\begin{lemma} \label{properties product transposed}
Given $ \mathbf{c}, \mathbf{d} \in \mathbb{F}_{q^m}^n $ and $ \mathbf{a}, \mathbf{b} \in \mathbb{F}_{q^m}^m $, and given $ C,D \in \mathbb{F}_q^{m \times n} $ and $ A,B \in \mathbb{F}_q^{m \times m} $, the following properties hold:
\begin{enumerate}
\item
$ M(\mathbf{c}(\mathbf{d})) = M(\mathbf{c}) M(\mathbf{d})^T $ and $ \mathbf{c}(\mathbf{d})^T = \mathbf{d}(\mathbf{c}) $.
\item
$ \mathbf{b} \cdot \mathbf{a}^T = \mathbf{b}^T \cdot \mathbf{a} $ and $ \langle B , A^T \rangle = \langle B^T, A \rangle $.
\item
$ (\mathbf{b} \star \mathbf{c}) \cdot \mathbf{d} = \mathbf{b} \cdot \mathbf{d}(\mathbf{c}) = \mathbf{b}^T \cdot \mathbf{c}(\mathbf{d}) = (\mathbf{b}^T \star \mathbf{d}) \cdot \mathbf{c} $.
\item
$ \langle BC , D \rangle = \langle B, DC^T \rangle = \langle B^T, CD^T \rangle = \langle B^T D, C \rangle $.
\item
$ \mathbf{c}(\mathbf{d}) = \mathbf{0} $ if, and only if, $ \mathbf{d}(\mathbf{c}) = \mathbf{0} $ if, and only if, $ {\rm RSupp}(\mathbf{c}) \subseteq {\rm RSupp}(\mathbf{d})^\perp $.
\item
$ CD^T = 0 $ if, and only if, $ DC^T = 0 $ if, and only if, $ {\rm Row}(C) \subseteq {\rm Row}(D)^\perp $.
\end{enumerate}
\end{lemma}
\begin{proof}
They are straightforward computations. For the first part of item 1, observe that
$$ \mathbf{c}(\mathbf{d}) = (\mathbf{c} \cdot \mathbf{d}_1, \mathbf{c} \cdot \mathbf{d}_2, \ldots, \mathbf{c} \cdot \mathbf{d}_m) = \sum_{i=1}^m \alpha_i (\mathbf{c}_i \cdot \mathbf{d}_1, \mathbf{c}_i \cdot \mathbf{d}_2, \ldots, \mathbf{c}_i \cdot \mathbf{d}_m). $$
Hence
$$ M(\mathbf{c}(\mathbf{d}))_{i,k}= \mathbf{c}_i \cdot \mathbf{d}_k = \sum_{j=1}^nc_{i,j}d_{k,j}=
\sum_{j=1}^n M(\mathbf{c})_{i,j}M(\mathbf{d})^T_{j,k}. $$
Therefore, $ M(\mathbf{c}(\mathbf{d})) =  M(\mathbf{c}) M(\mathbf{d})^T $.

For the first identity in Item 3, 
$$ (\mathbf{b} \star \mathbf{c}) \cdot \mathbf{d} =
\left( \sum_{i=1}^m b_i\mathbf{c}_i \right) \cdot \mathbf{d} =
\sum_{i=1}^m b_i(\mathbf{c}_i\cdot \mathbf{d}) =
\mathbf{b} \cdot \mathbf{d}(\mathbf{c}). $$
The first equivalence in Item 5 follows from item 1: $ \mathbf{c}(\mathbf{d})^T = \mathbf{d}(\mathbf{c}) $. Now, the second equivalence follows from the following chain of equivalences:
$$
\mathbf{c}(\mathbf{d}) =\mathbf{0} \Longleftrightarrow
\mathbf{c}_k \cdot \mathbf{d}_i, \forall i,k \Longleftrightarrow
{\rm RSupp} (\mathbf{c}) \subseteq {\rm RSupp} (\mathbf{d})^\perp .
$$
\end{proof}

\subsection{Using the extension inner product} \label{using extension}

Denote by $ \mathcal{D}^\perp $ the dual of an $ \mathbb{F}_{q^m} $-linear code $ \mathcal{D} \subseteq \mathbb{F}_{q^m}^n $ with respect to the extension product $ \cdot $. Fix $ \mathbb{F}_{q^m} $-linear codes $ \mathcal{A}, \mathcal{C} \subseteq \mathbb{F}_{q^m}^n $ and $ \mathcal{B} \subseteq \mathbb{F}_{q^m}^m $ such that $ \mathcal{B} \star \mathcal{A} \subseteq \mathcal{C}^\perp $, where $ \mathcal{B} \star \mathcal{A} $ is defined as
\begin{equation}
\mathcal{B} \star \mathcal{A} = \langle \{ \mathbf{b} \star \mathbf{a} \mid \mathbf{a} \in \mathcal{A}, \mathbf{b} \in \mathcal{B} \} \rangle_{\mathbb{F}_{q^m}}.
\label{A star B}
\end{equation}
In many cases, $ \mathcal{B} = \varphi(\mathcal{B}^\prime) $, where $ \varphi : \mathbb{F}_{q^m}^n \longrightarrow \mathbb{F}_{q^m}^m $ and $ \mathcal{B}^\prime \subseteq \mathbb{F}_{q^m}^n $ are both $ \mathbb{F}_{q^m} $-linear. In that case, we denote $ \mathcal{B}^\prime \star_\varphi \mathcal{A} = \varphi(\mathcal{B}^\prime) \star \mathcal{A} $.

Observe that, since $ \mathcal{B} $ is $ \mathbb{F}_{q^m} $-linear and $ \star $ is $ \mathbb{F}_{q^m} $-linear in the first component, it holds that $ \langle \{ \mathbf{b} \star \mathbf{a} \mid \mathbf{a} \in \mathcal{A}, \mathbf{b} \in \mathcal{B} \} \rangle_{\mathbb{F}_{q^m}} = \langle \{ \mathbf{b} \star \mathbf{a} \mid \mathbf{a} \in \mathcal{A}, \mathbf{b} \in \mathcal{B} \} \rangle_{\mathbb{F}_q} $.

We next compute generators of this space:

\begin{proposition} \label{generators ext}
If $ \mathbf{a}_1, \mathbf{a}_2, \ldots, \mathbf{a}_r $ generate $ \mathcal{A} $ and $ \mathbf{b}_1, \mathbf{b}_2, \ldots, \mathbf{b}_s $ generate $ \mathcal{B} $, as $ \mathbb{F}_{q^m} $-linear spaces, then the vectors
$$ \mathbf{b}_i \star (\alpha_l \mathbf{a}_j), $$
for $ 1 \leq i \leq s $, $ 1 \leq j \leq r $ and $ 1 \leq l \leq m $, generate $ \mathcal{B} \star \mathcal{A} $ as an $ \mathbb{F}_{q^m} $-linear space.
\end{proposition}

In the case $ \mathcal{B} = \varphi(\mathcal{B}^\prime) $ and $ \mathbf{b}^\prime_1, \mathbf{b}^\prime_2, \ldots, \mathbf{b}^\prime_s $ generate $ \mathcal{B}^\prime $ as an $ \mathbb{F}_{q^m} $-linear space, then the elements $ \mathbf{b}^\prime_i \star_\varphi (\alpha_l \mathbf{a}_j) $ generate $ \mathcal{B}^\prime \star_\varphi \mathcal{A} $ as an $ \mathbb{F}_{q^m} $-linear space.

Regarding the dimension of $ \mathcal{B} \star \mathcal{A} $ (or $ \mathcal{B} \star_\varphi \mathcal{A} $), that is, how many of the elements $ \mathbf{b}_i \star (\alpha_l \mathbf{a}_j) $ are linearly independent, the next example shows that any number may be possible in the case $ n \leq m $, where the previous proposition says that an upper bound in the general case is $ \min \{ \dim(\mathcal{A}) \dim(\mathcal{B}) m, n \} $:

\begin{example}
Assume that $ n \leq m $, fix $ 1 \leq t \leq n $, and define $ \mathbf{a} = (\alpha_1, \alpha_2, \ldots, \alpha_n) \in \mathbb{F}_{q^m}^n $ and $ \mathbf{b} = \mathbf{a} + \mathbf{a}^{[1]} + \cdots + \mathbf{a}^{[t-1]} \in \mathbb{F}_{q^m}^n $. Let $ \gamma \in \mathbb{F}_{q^m} $ be such that $ \gamma, \gamma^{[1]}, \ldots, \gamma^{[t-1]} $ are pairwise distinct, and write $ \gamma_i = \gamma^{[i]} $, for $ i = 0,1, \ldots, t-1 $. Let $ \mathcal{A} $ and $ \mathcal{B} $ be the $ \mathbb{F}_{q^m} $-linear spaces generated by $ \mathbf{a} $ and $ \mathbf{b} $, respectively. By Proposition \ref{evaluation properties}, item 2, we have that
$$ \mathbf{b} \star (\gamma^j \mathbf{a}) = \sum_{i=0}^{t-1} \mathbf{a}^{[i]} \star (\gamma^j \mathbf{a}) = \gamma_0^j \mathbf{a} + \gamma_1^j \mathbf{a}^{[1]} + \cdots + \gamma_{t-1}^j \mathbf{a}^{[t-1]} \in \mathcal{B} \star \mathcal{A}, $$
for $ j = 0,1,2, \ldots, t-1 $, and these elements are linearly independent over $ \mathbb{F}_{q^m} $, since the coefficients $ \gamma_i^j $ of the vectors $ \mathbf{a}^{[i]} $ form a Vandermonde matrix. Furthermore, $ \mathcal{B} \star \mathcal{A} $ is contained in the subspace generated by $ \mathbf{a}, \mathbf{a}^{[1]}, \ldots, \mathbf{a}^{[t-1]} $, hence they are equal. Therefore, $ \dim(\mathcal{A}) = \dim(\mathcal{B}) = 1 $, whereas $ \dim(\mathcal{B} \star \mathcal{A}) = t $.
\end{example}

Let $ \mathbf{d} \in \mathbb{F}_{q^m}^n $ and define
$$ \mathcal{K}(\mathbf{d}) = \{ \mathbf{a} \in \mathcal{A} \mid (\mathbf{b} \star \mathbf{a}) \cdot \mathbf{d} = 0, \forall \mathbf{b} \in \mathcal{B} \}. $$
Then $ \mathcal{K}(\mathbf{d}) $ is $ \mathbb{F}_q $-linear and the condition defining it may be verified just on a basis of $ \mathcal{B} $ as $ \mathbb{F}_{q^m} $-linear space. Observe that (precomputing the values $ \varphi(\mathbf{b}^\prime) $, where the vectors $ \mathbf{b}^\prime $ are in a basis of $ \mathcal{B}^\prime $, in the case $ \mathcal{B} = \varphi(\mathcal{B}^\prime) $), we can efficiently verify whether $ \mathbf{a} \in \mathcal{K}(\mathbf{d}) $. On the other hand, if $ \mathcal{L} \subseteq \mathbb{F}_q^n $ is a linear subspace, define
$$ \mathcal{A}(\mathcal{L}) = \{ \mathbf{a} \in \mathcal{A} \mid {\rm RSupp}(\mathbf{a}) \subseteq \mathcal{L}^\perp \}, $$
as in \cite{jurrius-pellikaan, slides}. We briefly connect this definition with the so-called rank-shortened codes in \cite[Definition 6]{similarities}, where $ \mathcal{A}_{\mathcal{L}^\perp} = \mathcal{A} \cap \mathcal{V}^\perp $ and $ \mathcal{V} = \mathcal{L} \otimes \mathbb{F}_{q^m} $ is defined as the $ \mathbb{F}_{q^m} $-linear vector space in $ \mathbb{F}_{q^m}^n $ generated by $ \mathcal{L} $:

\begin{lemma} \label{lemma shortened}
It holds that $ \mathcal{A}(\mathcal{L}) = \mathcal{A}_{\mathcal{L}^\perp} $. In particular, it is an $ \mathbb{F}_{q^m} $-linear space.
\end{lemma}
\begin{proof}
Fix a basis $ \mathbf{v}_1, \mathbf{v}_2, \ldots, \mathbf{v}_w $ of $ \mathcal{L} $. The result follows from the following chain of equivalent conditions
$$ {\rm RSupp}(\mathbf{a}) \in \mathcal{L}^\perp \Longleftrightarrow \mathbf{a}_i \in \mathcal{L}^\perp, \forall i \Longleftrightarrow \mathbf{a}_i \cdot \mathbf{v}_j = 0, \forall i,j \Longleftrightarrow \mathbf{a} \cdot \mathbf{v}_j = 0, \forall j \Longleftrightarrow \mathbf{a} \in \mathcal{V}^\perp. $$
\end{proof}

The following properties are the basic tools for the decoding algorithm of error correcting pairs:

\begin{proposition} \label{properties extension}
Let $ \mathbf{r} = \mathbf{c} + \mathbf{e} $, where $ \mathbf{c} \in \mathcal{C} $ and $ {\rm wt_R}(\mathbf{e}) \leq t $. Define also $ \mathcal{L} = {\rm RSupp}(\mathbf{e}) \subseteq \mathbb{F}_q^n $. The following properties hold:
\begin{enumerate}
\item
$ \mathcal{K}(\mathbf{r}) = \mathcal{K}(\mathbf{e}) $.
\item
$ \mathcal{A}(\mathcal{L}) \subseteq \mathcal{K}(\mathbf{e}) $.
\item
If $ t < d_{R}(\mathcal{B}^\perp) $, then $ \mathcal{A}(\mathcal{L}) = \mathcal{K}(\mathbf{e}) $. In this case, $ \mathcal{K}(\mathbf{e}) $ is $ \mathbb{F}_{q^m} $-linear.
\end{enumerate}
\end{proposition}
\begin{proof}
\begin{enumerate}
\item
It follows from $ \mathcal{B} \star \mathcal{A} \subseteq \mathcal{C}^\perp $.
\item
Let  $ \mathbf{a} \in A(L) $. It follows from Lemma \ref{properties product transposed} that
$\mathbf{e}(\mathbf{a})= \mathbf{0}$. Hence
$ (\mathbf{b} \star \mathbf{a}) \cdot \mathbf{e} = \mathbf{b} \cdot \mathbf{e}(\mathbf{a}) = 0 $, for all $ \mathbf{b} \in B $. Thus $ \mathbf{a} \in K(\mathbf{e})$.
\item
By the previous item, we only need to prove that $ \mathcal{K}(\mathbf{e}) \subseteq \mathcal{A}(\mathcal{L}) $.

Let $ \mathbf{a} \in \mathcal{K}(\mathbf{e}) $. It follows from Lemma \ref{properties product transposed} that $ \mathbf{e}(\mathbf{a}) \in \mathcal{B}^\perp $. Moreover, since $ M(\mathbf{e}(\mathbf{a})) = M(\mathbf{e}) M(\mathbf{a})^T $ by the same lemma, it holds that $ {\rm wt_R}(\mathbf{e}(\mathbf{a})) \leq {\rm wt_R}(\mathbf{e}) \leq t $.

Since $ t < d_{R}(\mathcal{B}^\perp) $, it follows that $ \mathbf{e}(\mathbf{a}) = \mathbf{0} $ or, in other words, $ \mathbf{a}_i \cdot \mathbf{e} = 0 $, which implies that $ \mathbf{a}_i \in \mathcal{L}^\perp $, for all $ i =1,2, \ldots, m $, and therefore, $ {\rm RSupp}(\mathbf{a}) \subseteq \mathcal{L}^\perp $. \\
\end{enumerate}
\end{proof}

We now come to the definition of $ t $-rank error-correcting pairs of type I, where we use the extension inner product $ \cdot $. 

\begin{definition} \label{RECP type I}
The pair $ (\mathcal{A},\mathcal{B}) $ is called a $ t $-rank error-correcting pair ($ t $-RECP) of type I for $ \mathcal{C} $ if the following properties hold:
\begin{enumerate}
\item
$ \mathcal{B} \star \mathcal{A} \subseteq \mathcal{C}^\perp $. 
\item
$ \dim(\mathcal{A}) > t $.
\item
$ d_R(\mathcal{B}^\perp) > t $.
\item
$ d_R(\mathcal{A}) + d_R(\mathcal{C}) > n $.
\end{enumerate}
If $ \mathcal{B} = \varphi(\mathcal{B}^\prime) $, where $ \varphi $ and $ \mathcal{B}^\prime \subseteq \mathbb{F}_{q^m}^n $ are $ \mathbb{F}_{q^m} $-linear, we say that $ (\mathcal{A}, \mathcal{B}^\prime) $ is a $ t $-RECP of type I for $ \varphi $ and $ \mathcal{C} $, and if $ \varphi = \varphi_n $, we will call it simply a $ t $-RECP of type I for $ \mathcal{C} $.
\end{definition}

In order to describe a decoding algorithm for $ \mathcal{C} $ using $ (\mathcal{A},\mathcal{B}) $, we will need \cite[Proposition 16]{similarities}, slightly modified (the proof is the same), which basically states that error correction is equivalent to erasure correction if the rank support of the error is known:

\begin{lemma} \label{error erasure}
Assume that $ \mathbf{c} \in \mathcal{C} $ and $ \mathbf{r} = \mathbf{c} + \mathbf{e} $, where $ {\rm RSupp}(\mathbf{e}) \subseteq \mathcal{L} $ and $ \dim(\mathcal{L}) < d_R(\mathcal{C}) $. Then, $ \mathbf{c} $ is the only vector in $ \mathcal{C} $ such that $ {\rm RSupp}(\mathbf{r} - \mathbf{c}) \subseteq \mathcal{L} $.

Moreover, if $ G $ is a generator matrix of $ \mathcal{L}^\perp $, then $ \mathbf{c} $ is the unique solution in $ \mathcal{C} $ of the system of equations $ \mathbf{r}G^T = \mathbf{x}G^T $, where $ \mathbf{x} $ is the unknown vector. And if $ H $ is a parity check matrix for $ \mathcal{C} $ over $ \mathbb{F}_{q^m} $, then $ \mathbf{e} $ is the unique solution to the system $ \mathbf{r}H^T = \mathbf{x}H^T $ with $ {\rm RSupp} (\mathbf{x}) \subseteq \mathcal{L} $.
\end{lemma}

Now we present, in the proof of the following theorem, a decoding algorithm for $ \mathcal{C} $ using $ (\mathcal{A}, \mathcal{B}) $.

\begin{theorem} \label{decoding type I}
If $ (\mathcal{A}, \mathcal{B}) $ is a $ t $-RECP of type I for $ \mathcal{C} $, then $ \mathcal{C} $ verifies that $ d_R(\mathcal{C}) \geq 2t + 1 $ and admits a decoding algorithm able to correct errors $ \mathbf{e} $ with $ {\rm wt_R}(\mathbf{e}) \leq t $ of complexity $ O(n^3) $ over the field $ \mathbb{F}_{q^m} $.
\end{theorem}
\begin{proof}
We will explicitly describe the decoding algorithm. As a consequence, we will derive that $ d_R(\mathcal{C}) \geq 2t + 1 $. Assume that the received codeword is $ \mathbf{r} = \mathbf{c} + \mathbf{e} $, with $ \mathbf{c} \in \mathcal{C} $, $ {\rm RSupp}(\mathbf{e}) = \mathcal{L} $ and $ \dim(\mathcal{L}) \leq t $.

Compute the space $ \mathcal{K}(\mathbf{r}) $, which is equal to $ \mathcal{K}(\mathbf{e}) $ by the first condition of $ t $-RECP and Proposition \ref{properties extension}, item 1. Observe that $ \mathcal{K}(\mathbf{r}) $ can be described by a system of $ O(n) $ linear equations by Proposition \ref{generators ext}.

By the third condition of $ t $-RECP and Proposition \ref{properties extension}, we have that $ \mathcal{A}(\mathcal{L}) = \mathcal{K}(\mathbf{e}) = \mathcal{K}(\mathbf{r}) $. Therefore, we have computed the space $ \mathcal{A}(\mathcal{L}) $.

By the second condition of $ t $-RECP and Lemma \ref{lemma shortened}, we have that $ \mathcal{A}(\mathcal{L}) = \mathcal{A} \cap \mathcal{V}^\perp \neq 0 $, and therefore we may take a nonzero $ \mathbf{a} \in \mathcal{A}(\mathcal{L}) $. Define $ \mathcal{L}^\prime = {\rm RSupp}(\mathbf{a})^\perp $. Since $ \mathbf{a} \in \mathcal{A}(\mathcal{L}) $, we have that $ \mathcal{L} \subseteq \mathcal{L}^\prime $.

Now, by the fourth condition of $ t $-RECP, we have that
$$ \dim(\mathcal{L}^\prime) = n - {\rm wt_R}(\mathbf{a}) \leq n - d_R(\mathcal{A}) < d_R(\mathcal{C}). $$
Hence, by the previous lemma, we may compute $ \mathbf{e} $ or $ \mathbf{c} $ by solving a system of linear equations using a generator matrix $ G $ of $ \mathcal{L}^{\prime \perp} $, or a parity check matrix $ H $ of $ \mathcal{C} $, respectively. This has complexity $ O(n^3) $ over $ \mathbb{F}_{q^m} $.

Finally, assume that $ d_R(\mathcal{C}) \leq 2t $ and take two different vectors $ \mathbf{c}, \mathbf{c}^\prime \in \mathcal{C} $ and $ \mathbf{e}, \mathbf{e}^\prime \in \mathbb{F}_{q^m}^n $ such that $ \mathbf{r} = \mathbf{c} + \mathbf{e} = \mathbf{c}^\prime + \mathbf{e}^\prime $ and $ {\rm wt_R}(\mathbf{e}), {\rm wt_R}(\mathbf{e}^\prime) \leq t $. The previous algorithm gives as output both vectors $ \mathbf{e} $ and $ \mathbf{e}^\prime $, but the output is unique, hence $ \mathbf{e} = \mathbf{e}^\prime $. This implies that $ \mathbf{c} = \mathbf{c}^\prime $, contradicting the hypothesis. Therefore, $ d_R(\mathcal{C}) \geq 2t+1 $.
\end{proof}

If $ m=n $, then the order of complexity over $ \mathbb{F}_q $ increases, although it still is polynomial in $ n $. On the other hand, if $ m $ is considerably smaller than $ n $, then the complexity is $ O(n^3) $ also over $ \mathbb{F}_q $.

Gabidulin codes \cite{gabidulin} have decoding algorithms of cubic complexity (see for instance \cite{gabidulin}), and an algorithm of quadratic complexity was obtained in \cite{loidreau-gabidulin}. As we will see in Section \ref{codes with}, the previous decoding algorithm may be applied to a wider variety of rank-metric codes.

\begin{remark} \label{error-locating I}
Observe that, from the proof of the previous theorem, if the pair $ (\mathcal{A},\mathcal{B}) $ satisfies the first three properties in Definition \ref{RECP type I}, then we may use it to find a subspace $ \mathcal{L}^\prime \subseteq \mathbb{F}_q^n $ that contains the rank support of the error vector. 

Therefore, we say in this case that $ (\mathcal{A},\mathcal{B}) $ is a $ t $-rank error-locating pair of type I for $ \mathcal{C} $.
\end{remark}

\subsection{Using the base inner product} \label{using base}

Now we turn to the case where we use the base inner product $ \langle , \rangle $. We will denote by $ \mathcal{D}^* $ the dual of an $ \mathbb{F}_q $-linear code $ \mathcal{D} \subseteq \mathbb{F}_{q}^{m \times n} $ with respect to $ \langle , \rangle $.

We will use the same notation as in the previous subsection, although now $ \mathcal{A}, \mathcal{C} \subseteq \mathbb{F}_{q}^{m \times n} $ and $ \mathcal{B} \subseteq \mathbb{F}_{q}^{m \times m} $ are $ \mathbb{F}_q $-linear, and $ \mathcal{B} \mathcal{A} \subseteq \mathcal{C}^* $, where
\begin{equation}
\mathcal{B} \mathcal{A} = \langle \{ B A \mid A \in \mathcal{A}, B \in \mathcal{B} \} \rangle_{\mathbb{F}_q}.
\label{A star B base}
\end{equation}
Observe that $ M( \mathcal{B}^\prime \star \mathcal{A}^\prime) = M(\mathcal{B}^\prime)M(\mathcal{A}^\prime) $, if $ \mathcal{A}^\prime, \mathcal{B}^\prime \subseteq \mathbb{F}_{q^m}^n $ are $ \mathbb{F}_q $-linear spaces, by Proposition \ref{connection vector products}. Generators of this space are now simpler to compute:

\begin{proposition} \label{generators base}
If $ A_1, A_2, \ldots, A_r $ generate $ \mathcal{A} $ and $ B_1, B_2, \ldots, B_s $ generate $ \mathcal{B} $, as $ \mathbb{F}_q $-linear spaces, then the matrices
$$ B_i A_j, $$
for $ 1 \leq i \leq s $ and $ 1 \leq j \leq r $, generate $ \mathcal{B} \mathcal{A} $ as an $ \mathbb{F}_q $-linear space.
\end{proposition}

Let $ D \in \mathbb{F}_q^{m \times n} $ and define
$$ \mathcal{K}(D) = \{ A \in \mathcal{A} \mid \langle B A , D \rangle = 0, \forall B \in \mathcal{B} \}. $$
Then $ \mathcal{K}(D) $ is again $ \mathbb{F}_q $-linear and the condition may be verified just on a basis of $ \mathcal{B} $ as $ \mathbb{F}_q $-linear space. On the other hand, if $ \mathcal{L} \subseteq \mathbb{F}_q^n $ is a linear subspace, we define in the same way
$$ \mathcal{A}(\mathcal{L}) = \{ A \in \mathcal{A} \mid {\rm Row}(A) \subseteq \mathcal{L}^\perp \}, $$
which is $ \mathbb{F}_q $-linear (recall that we use the classical product $ \cdot $ in $ \mathbb{F}_q^n $), since we still have that $ M^{-1}( \mathcal{A}(\mathcal{L}) ) = M^{-1}( \mathcal{A}) \cap \mathcal{V}^\perp $, $ \mathcal{V} = \mathcal{L} \otimes \mathbb{F}_{q^m} $.

The following properties still hold:

\begin{proposition} \label{properties base}
Let $ R = C + E $, where $ C \in \mathcal{C} $ and $ {\rm Rk}(E) \leq t $. Define also $ \mathcal{L} = {\rm Row}(E) \subseteq \mathbb{F}_q^n $. Then
\begin{enumerate}
\item
$ \mathcal{K}(R) = \mathcal{K}(E) $.
\item
$ \mathcal{A}(\mathcal{L}) \subseteq \mathcal{K}(E) $.
\item
If $ t < d_{R}(\mathcal{B}^*) $, then $ \mathcal{A}(\mathcal{L}) = \mathcal{K}(E) $.
\end{enumerate}
\end{proposition}
\begin{proof}
\begin{enumerate}
\item
It also follows from $ \mathcal{B} \mathcal{A} \subseteq \mathcal{C}^* $.
\item
Take $ A \in \mathcal{A}(\mathcal{L}) $. Hence by Lemma \ref{properties product transposed}, it holds that $ EA^T = 0 $, since $ {\rm Row}(E) = \mathcal{L} $. Therefore, for every $ B \in \mathcal{B} $, we have that
$$ \langle B A , E \rangle = \langle B, EA^T \rangle = 0, $$
by Lemma \ref{properties product transposed}. Then item 2 follows.
\item
By the previous item, we only need to prove that $ \mathcal{K}(E) \subseteq \mathcal{A}(\mathcal{L}) $.

Let $ A \in \mathcal{K}(E) $. It follows from Lemma \ref{properties product transposed} that $ EA^T \in B^* $. Moreover, it holds that $ {\rm Rk}(EA^T) \leq {\rm Rk}(E) \leq t $. Since $ t < d_{R}(\mathcal{B}^*) $, it follows that $ EA^T = 0 $, which implies that $ {\rm Row}(A) \in \mathcal{L}^\perp $. \\
\end{enumerate}
\end{proof}

We now define $ t $-rank error-correcting pairs of type II, where we use the base product $ \langle , \rangle $, in contrast with the $ t $-RECP of last subsection.

\begin{definition} \label{RECP type II}
The pair $ (\mathcal{A},\mathcal{B}) $ is called a $ t  $-rank error-correcting pair ($ t $-RECP) of type II for $ \mathcal{C} $ if the following properties hold:
\begin{enumerate}
\item
$ \mathcal{B} \mathcal{A} \subseteq \mathcal{C}^* $.
\item
$ \dim_{\mathbb{F}_q}(\mathcal{A}) > mt $.
\item
$ d_R(\mathcal{B}^*) > t $.
\item
$ d_R(\mathcal{A}) + d_R(\mathcal{C}) > n $.
\end{enumerate}
\end{definition}

The same decoding algorithm, with the corresponding modifications, works in this case with polynomial complexity:

\begin{theorem} \label{decoding type II}
If $ (\mathcal{A},\mathcal{B}) $ is a $ t $-RECP of type II for $ \mathcal{C} $, then $ \mathcal{C} $ satisfies that $ d_R(\mathcal{C}) \geq 2t+1 $ and admits a decoding algorithm able to correct errors $ E $ with $ {\rm Rk}(E) \leq t $ with polynomial complexity in $ (m,n) $ over the field $ \mathbb{F}_q $. 
\end{theorem}
\begin{proof}
The proof is the same as in Theorem \ref{decoding type I}, with the corresponding modifications. Note that in this case, if $ \mathcal{L} = {\rm Row}(E) $ and $ \mathcal{V} = \mathcal{L} \otimes \mathbb{F}_{q^m} $, then $ \dim_{\mathbb{F}_q}(\mathcal{V}) = m \dim(\mathcal{L}) \leq mt $. On the other hand, $ M^{-1}(\mathcal{A}(\mathcal{L})) = M^{-1}(\mathcal{A}) \cap \mathcal{V} $, as in the previous subsection. Hence the condition $ \dim_{\mathbb{F}_q}(\mathcal{A}) > mt $ ensures that $ \mathcal{A}(\mathcal{L}) \neq 0 $.
\end{proof}

\begin{remark} \label{error-locating II}
As in Remark \ref{error-locating I}, if the pair $ (\mathcal{A},\mathcal{B}) $ satisfies the first three properties in Definition \ref{RECP type II}, then we may use it to find a subspace $ \mathcal{L}^\prime \subseteq \mathbb{F}_q^n $ that contains the rank support of the error vector. We say in this case that $ (\mathcal{A},\mathcal{B}) $ is a $ t $-rank error-locating pair of type II for $ \mathcal{C} $.
\end{remark}

\section{The connection between the two types of RECPs} \label{connections}

So far we have three types of error-correcting pairs: classical ECPs for linear codes in $ \mathbb{F}_q^n $ that correct errors in the Hamming metric, ECPs for $ \mathbb{F}_{q^m} $-linear codes in $ \mathbb{F}_{q^m}^n $ (RECPs of type I), and ECPs for general $ \mathbb{F}_q $-linear codes in $ \mathbb{F}_{q^m}^n $ or $ \mathbb{F}_q^{m \times n} $ (RECPs of type II), where the two latter types correct errors in the rank metric. In this section we will see that RECPs of type II generalize RECPs of type I. In Section \ref{codes with} we will see that, in some way, RECPs of type II also generalize ECPs for the Hamming metric.

We will need the following:

\begin{definition}
Given the basis $ \alpha_1, \alpha_2, \ldots, \alpha_m $ of $ \mathbb{F}_{q^m} $ over $ \mathbb{F}_q $, we say that it is orthogonal (or dual) to another basis $ \alpha^\prime_1, \alpha^\prime_2, \ldots, \alpha^\prime_m $ if
$$ {\rm Tr}(\alpha_i \alpha^\prime_j) = \delta_{i,j}, $$
for all $ i,j = 1,2, \ldots, m $. Here, $ {\rm Tr} $ denotes the trace of the extension $ \mathbb{F}_q \subseteq \mathbb{F}_{q^m} $.
\end{definition}

It is well-known that, for a given basis $ \alpha_1, \alpha_2, \ldots, \alpha_m $, there exists a unique orthogonal basis (see for instance the discussion after \cite[Definition 2.50]{lidl}). We will denote it as in the previous definition: $ \alpha^\prime_1, \alpha^\prime_2, \ldots, \alpha^\prime_m $. In particular, the dual basis of $ \alpha^\prime_1, \alpha^\prime_2, \ldots, \alpha^\prime_m $ is $ \alpha_1, \alpha_2, \ldots, \alpha_m $.

Now denote by $ M_\alpha, M_{\alpha^\prime} : \mathbb{F}_{q^m}^n \longrightarrow \mathbb{F}_q^{m \times n} $ the matrix representation maps associated to the previous bases, respectively. The following lemma is \cite[Theorem 21]{ravagnani}:

\begin{lemma} \label{duality ravagnani}
Given an $ \mathbb{F}_{q^m} $-linear code $ \mathcal{C} \subseteq \mathbb{F}_{q^m}^n $, it holds that
$$ M_{\alpha^\prime}(\mathcal{C}^\perp) = M_\alpha(\mathcal{C})^*. $$
\end{lemma}

On the other hand, we have the following:

\begin{lemma} \label{duality distance}
For every $ \mathbb{F}_{q^m} $-linear code $ \mathcal{D} \subseteq \mathbb{F}_{q^m}^n $, it holds that
$$ d_R(\mathcal{D}^\perp) = d_R(M_\alpha(\mathcal{D})^*) = d_R(M_{\alpha^\prime}(\mathcal{D})^*). $$
\end{lemma}
\begin{proof}
It follows from the fact that $ d_R(\mathcal{D}^\perp) = d_R(M_{\alpha^\prime}(\mathcal{D}^\perp)) = d_R(M_\alpha(\mathcal{D})^*) $, and analogously interchanging the roles of $ \alpha $ and $ \alpha^\prime $.
\end{proof}

Therefore, we may now prove that RECPs of type II generalize RECPs of type I:

\begin{theorem} \label{RECP I admit RECP II}
Take $ \mathbb{F}_{q^m} $-linear codes $ \mathcal{A},\mathcal{C} \subseteq \mathbb{F}_{q^m}^n $ and $ \mathcal{B} \subseteq \mathbb{F}_{q^m}^m $. If $ (\mathcal{A},\mathcal{B}) $ is a $ t $-RECP of type I for $ \mathcal{C} $ (in the basis $ \alpha $), then $ (M_\alpha(\mathcal{A}),M_\alpha(\mathcal{B})) $ is a $ t $-RECP of type II for $ M_{\alpha^\prime}(\mathcal{C}) $.
\end{theorem}
\begin{proof}
Using Lemma \ref{duality ravagnani} and Proposition \ref{connection vector products}, we obtain that
$$ M_\alpha(\mathcal{B}) M_\alpha(\mathcal{A}) = M_\alpha(\mathcal{B} \star \mathcal{A}) \subseteq M_\alpha(\mathcal{C}^\perp) = M_{\alpha^\prime}(\mathcal{C})^*, $$
and the first condition is satisfied.

The second condition follows from the fact that $ \dim_{\mathbb{F}_q}(\mathcal{A}) = m \dim_{\mathbb{F}_{q^m}}(\mathcal{A}) $, and $ M_\alpha $ is an $ \mathbb{F}_q $-linear vector space isomorphism.

Finally, the third condition follows from Lemma \ref{duality distance} and the fourth condition remains unchanged. Hence the result follows.
\end{proof}

Observe that in the same way, $ t $-rank error-locating pairs of type II generalize $ t $-rank error-locating pairs of type I.

\section{MRD codes and bounds on the minimum rank distance}

In this section we will give bounds on the minimum rank distance of codes that follow from the properties of rank error-correcting pairs, in a similar way to the bounds in \cite{on-the-existence}. We will also see that, in some cases, MRD conditions on two of the codes imply that the third is also MRD.

We will fix $ \mathbb{F}_{q} $-linear codes $ \mathcal{A},\mathcal{C} \subseteq \mathbb{F}_{q}^{m \times n} $ and $ \mathcal{B} \subseteq \mathbb{F}_{q}^{m \times m} $. Due to Lemmas \ref{duality ravagnani} and \ref{duality distance}, and Proposition \ref{connection vector products}, the results in this section may be directly translated into results where we consider the ``extension'' inner product $ \cdot $ and $ \mathbb{F}_{q^m} $-linear codes in $ \mathbb{F}_{q^m}^n $.

We will make use of the following consequence of the Singleton bound:

\begin{lemma}
For every $ \mathbb{F}_q $-linear code $ \mathcal{D} \subseteq \mathbb{F}_q^{m \times n} $ it holds that
$$ d_R(\mathcal{D}) + d_R(\mathcal{D}^*) \leq n + 2. $$
\end{lemma}
\begin{proof}
The Singleton bound implies that
$$ \dim_{\mathbb{F}_q}(\mathcal{D})/m \leq n - d_R(\mathcal{D}) + 1, \quad \textrm{and} \quad \dim_{\mathbb{F}_q}(\mathcal{D}^*)/m \leq n - d_R(\mathcal{D}^*) + 1. $$
Adding both inequalities up and using that $ \dim_{\mathbb{F}_q}(\mathcal{D}) + \dim_{\mathbb{F}_q}(\mathcal{D}^*) = mn $, the result follows.
\end{proof}

\begin{proposition} \label{first bound}
Assume that $ \mathcal{B} \mathcal{A} \subseteq \mathcal{C}^* $. If $ d_R(\mathcal{A}^*) > a > 0 $ and $ d_R(\mathcal{B}^*) > b > 0 $, then $ d_R(\mathcal{C}) \geq a+b $.
\end{proposition}
\begin{proof}
Take $ C \in \mathcal{C} $ and $ A \in \mathcal{A} $, and define $ \mathcal{L} = {\rm Row}(C) \subseteq \mathbb{F}_q^n $. By Lemma \ref{properties product transposed}, we have that 
$$ 0 = \langle B  A , C \rangle = \langle B^T , AC^T \rangle, $$
for all $ B \in \mathcal{B} $ and all $ A \in \mathcal{A} $, which means that the $ \mathbb{F}_q $-linear space $ \mathcal{A}(C) = \{ AC^T \mid A \in \mathcal{A} \} \subseteq (\mathcal{B}^T)^* $, and hence $ d_R(\mathcal{A}(C)) > b $.

Let $ G $ be a $ t \times n $ generator matrix over $ \mathbb{F}_q $ of $ \mathcal{L} $ (where $ t = {\rm Rk}(C) $). Taking a subset of rows of $ C $ that generate $ \mathcal{L} $, we see that $ \mathcal{A}(C) $ is $ \mathbb{F}_q $-linearly isomorphic and rank-metric equivalent to $ \mathcal{A}_1 = \{ A G^T \mid A \in \mathcal{A} \} \subseteq \mathbb{F}_{q}^{m \times t} $. Take $ D \in \mathcal{A}_1^* $. For every $ A \in \mathcal{A} $, it holds that
$$ \langle A, D G \rangle = \langle A G^T, D \rangle = 0, $$
by Lemma \ref{properties product transposed}. Therefore, $ D G \in \mathcal{A}^* $. Moreover, $ {\rm Rk}(D) = {\rm Rk}(D G) $ since $ G $ is full rank, and hence $ {\rm Rk}(D) > a $. Therefore, $ d_R(\mathcal{A}_1^*) > a $. Together with $ d_R(\mathcal{A}_1) > b $ and the previous lemma, we obtain that
$$ a + 1 + b + 1 \leq d_R(\mathcal{A}_1) + d_R(\mathcal{A}_1^*) \leq t + 2, $$
that is, $ t \geq a + b $, and the result follows.
\end{proof}

We obtain the following corollary on MRD codes:

\begin{corollary}
Assume that $ n \leq m $ (otherwise, take transposed matrices), $ d_R(\mathcal{A}) = n - t $, $ \dim_{\mathbb{F}_q}(\mathcal{A}) = m(t+1) $, $ d_R(\mathcal{B}) = m - t + 1 $ and $ \dim_{\mathbb{F}_q}(\mathcal{B}) = mt $. Then, for all $ \mathcal{D} \subseteq (\mathcal{B} \mathcal{A})^* $, it holds that $ d_R(\mathcal{D}) \geq 2t+1 $ and $ (\mathcal{A}, \mathcal{B}) $ is a $ t $-RECP of type II for $ \mathcal{D} $.
\end{corollary}
\begin{proof}
$ \mathcal{A} $ and $ \mathcal{B} $ are MRD codes, since their minimum rank distance attains the Singleton bound. By \cite[Theorem 5.5]{delsartebilinear} (see also \cite[Corollary 41]{ravagnani}), $ \mathcal{A}^* $ and $ \mathcal{B}^* $ are also MRD, which implies that
$$ d_R(\mathcal{A}^*) > t+1, \quad \textrm{and} \quad d_R(\mathcal{B}^*) > t. $$
By the previous proposition, it holds that $ d_R(\mathcal{D}) \geq 2t + 1 $. We see that the properties of RECPs of type II are satisfied, and the result follows.
\end{proof}

Now we obtain bounds on $ d_R(\mathcal{A}) $ from bounds on $ d_R(\mathcal{B}^*) $ and $ d_R(\mathcal{C}^*) $:

\begin{proposition}
Assume that $ \mathcal{B} \mathcal{A} \subseteq \mathcal{C}^* $. If $ d_R(\mathcal{B}^*) > b > 0 $ and $ d_R(\mathcal{C}^*) > c > 0 $, then $ d_R(\mathcal{A}) \geq b+c $.
\end{proposition}
\begin{proof}
The proof is analogous to the proof of Proposition \ref{first bound}. In this case, we fix $ A \in \mathcal{A} $, with $ \mathcal{L} = {\rm Row}(A) $, $ t = {\rm Rk}(A) $, and consider $ A(\mathcal{C}) = \{ AC^T \mid C \in \mathcal{C} \} $. The rest of the proof follows the same lines, interchanging the roles of $ \mathcal{A} $ and $ \mathcal{C} $, and using the fact that $ \langle BA , C \rangle = \langle B^T C , A \rangle $, from Lemma \ref{properties product transposed}, and $ d_R(\mathcal{B}^*) = d_R((\mathcal{B}^T)^*) $.
\end{proof}

Again, we may give the following corollary on MRD codes:

\begin{corollary}
Assume that $\mathcal{B} \mathcal{A} \subseteq \mathcal{C}^* $ and $ n \leq m $. If $ d_R(\mathcal{C}) = 2t+1 $, $ \dim_{\mathbb{F}_q}(\mathcal{C}) = m(n - 2t) $ and $ (\mathcal{A}, \mathcal{B}) $ is a $ t $-RECP of type II for $ \mathcal{C} $, then $ d_R(\mathcal{A}) \geq n -t $ and $ mt < \dim_{\mathbb{F}_q}(\mathcal{A}) \leq m (t+1) $. If $ \dim_{\mathbb{F}_q}(\mathcal{A}) $ is a multiple of $ m $ (in particular, if $ M^{-1}(\mathcal{A}) $ is $ \mathbb{F}_{q^m} $-linear), then $ \mathcal{A} $ is MRD.
\end{corollary}
\begin{proof}
By the properties of RECPs of type II, we have that $ d_R(\mathcal{B}^*) > t $, and since $ \mathcal{C} $ is MRD, then $ \mathcal{C}^* $ is also MRD and we have that $ d_R(\mathcal{C}^*) = n - 2t + 1 $. Therefore, $ d_R(\mathcal{A}) \geq n - t $ by the previous proposition. By the properties of RECPs of type II, $ \dim_{\mathbb{F}_q}(\mathcal{A}) > mt $, and we are done. The last statement follows from the Singleton bound for $ \mathcal{A} $.
\end{proof}

We now turn to a bound analogous to \cite[Proposition 3.1]{on-the-existence}. The BCH bound on the minimum Hamming distance of cyclic codes is generalized by the Hartmann-Tzeng bounds \cite{hartmann} and further generalized by the Roos bound \cite{roos1, roos2}. The next proposition is the rank-metric equivalent of the Roos bound \cite{roos1, roos2} for the Hamming metric, as mentioned in \cite[Proposition 5]{pellikaan-duursma}.

\begin{proposition} \label{roos}
Assume the following properties for $ a,b > 0 $:  
$$ {\rm (1)} \mathcal{B} \mathcal{A} \subseteq \mathcal{C}^*, \quad {\rm (2)} \dim_{\mathbb{F}_q}(\mathcal{A}) > ma, \quad {\rm (3)} d_R(\mathcal{B}^*) >b, $$
$$ {\rm (4)} d_R(\mathcal{A}) + a + b > n, \quad \textrm{and} \quad {\rm (5)} d_R(\mathcal{A}^*) > 1. $$
Then it holds that $ d_R(\mathcal{C}) > a+b $.
\end{proposition}
\begin{proof}
Take $ C \in \mathcal{C} $ and let $ \mathcal{L} = {\rm Row}(C) \subseteq \mathbb{F}_q^n $ and $ t = {\rm Rk}(C) $. Conditions (1), (3) and (5) imply that $ t > b $ by Proposition \ref{first bound}. 

Assume that $ b < t \leq a+b $. Take linear subspaces $ \mathcal{L}_{-}, \mathcal{L}_{+}, \mathcal{U} \subseteq \mathbb{F}_q^n $ such that $ \mathcal{L}_{-} \subseteq \mathcal{L} \subseteq \mathcal{L}_{+} $, $ \mathcal{L}_{+} = \mathcal{U} \oplus \mathcal{L}_{-} $, $ b = \dim(\mathcal{L}_{-}) $ and $ a+b = \dim(\mathcal{L}_{+}) $. Since $ m \dim(\mathcal{U}) = ma < \dim_{\mathbb{F}_q}(\mathcal{A}) $ by condition (2), we have that $ \mathcal{A}(\mathcal{U}) \neq 0 $, and therefore there exists a non-zero $ A \in \mathcal{A} $ with $ {\rm Row}(A) \subseteq \mathcal{U}^\perp $.

It holds that every row in $ C $ is in $ \mathcal{L}_{+} $. Since the rows in $ A $ are in $ \mathcal{U}^\perp $, it holds that $ AC^T = A N^T $, where $ N $ is obtained from $ C $ by substituting every row by its projection from $ \mathcal{U} \oplus \mathcal{L}_{-} $ to $ \mathcal{L}_{-} $.

Therefore $ {\rm Rk}(AC^T) \leq {\rm Rk}(N) \leq \dim(\mathcal{L}_{-}) = b $, but $ AC^T \in (\mathcal{B}^T)^* $ by condition (1) and Lemma \ref{properties product transposed}, and hence $ AC^T = 0 $ by condition (3). This means that $ {\rm Row}(A) \subseteq \mathcal{L}_{-}^\perp \cap \mathcal{U}^\perp = \mathcal{L}_{+}^\perp  $. Thus, $ {\rm Rk}(A) \leq n - a - b < d_R(\mathcal{A}) $, which is absurd by condition (4), since $ A \neq 0 $. We conclude that $ t > a+b $ and we are done.
\end{proof}

Taking $ a=b=t $ for some $ t > 0 $, where $ a $ and $ b $ are as in the previous proof, we obtain the following particular case:

\begin{corollary} \label{stronger RECPs}
For all $ \mathbb{F}_q $-linear codes $ \mathcal{D} \subseteq (\mathcal{B} \mathcal{A})^* $ such that $ \dim_{\mathbb{F}_q}(\mathcal{A}) > mt $, $ d_R(\mathcal{B}^*) > t $, $ d_R(\mathcal{A}) > n - 2t $ and $ d_R(\mathcal{A}^*) > 1 $, it holds that $ d_R(\mathcal{D}) \geq 2t+1 $ and $ (\mathcal{A}, \mathcal{B}) $ is a $ t $-RECP of type II for $ \mathcal{D} $. 
\end{corollary}

Observe that the previous result states that, if some conditions on $ \mathcal{A} $ and $ \mathcal{B} $ hold, then they form a $ t $-RECP of type II for all $ \mathbb{F}_q $-linear codes contained in $ (\mathcal{B} \mathcal{A})^* $. That is, we have found a $ t $-rank error-correcting algorithm for all $ \mathbb{F}_q $-linear subcodes of $ (\mathcal{B} \mathcal{A})^* $.

\section{Some codes with a $ t $-RECP} \label{codes with}

In this section, we study families of codes that admit a $ t $-RECP of some type.

\subsection{Hamming-metric codes with ECPs} \label{subsection Hamming}

Take $ \mathbb{F}_q $-linear codes $ \mathcal{A}, \mathcal{B}, \mathcal{C} \subseteq \mathbb{F}_q^n $ such that $ (\mathcal{A}, \mathcal{B}) $ is a $ t $-ECP for $ \mathcal{C} $ in the Hamming metric. We will see that the algorithm presented in Theorem \ref{decoding type II} is actually an extension of the decoding algorithm in the Hamming metric using $ t $-ECPs \cite{on-ECP, pellikaan-error-location}. We observe the following (recall the definition of $ D $ in (\ref{diagonal})):

\begin{remark} \label{products extension Hamming}
For all $ \mathbf{a}, \mathbf{b} \in \mathbb{F}_q^n $, it holds that
$$ \mathbf{a} \cdot \mathbf{b} = \langle D(\mathbf{a}), D(\mathbf{b}) \rangle. $$
Moreover, it holds that
$$ D(\mathcal{B})  D(\mathcal{A}) \subseteq D(\mathcal{C})^*. $$
\end{remark}

Therefore, from the previous remark and the properties of $ D $, the $ \mathbb{F}_q $-linear codes $ D(\mathcal{A}), D(\mathcal{B}), D(\mathcal{C}) \subseteq \mathbb{F}_{q}^{n \times n} $ satisfy the following conditions:
\begin{enumerate}
\item
$ D(\mathcal{B})  D(\mathcal{A}) \subseteq D(\mathcal{C})^* $.
\item
$ \dim_{\mathbb{F}_q}(D(\mathcal{A})) > t $.
\item
$ d_R(D(\mathcal{B})^*) = 1 $.
\item
$ d_R(D(\mathcal{A})) + d_R(D(\mathcal{C})) > n $.
\end{enumerate}
That is, $ (D(\mathcal{A}), D(\mathcal{B})) $ satisfy the same conditions as $ t $-RECPs of type II for $ D(\mathcal{C}) $, except that conditions 2 and 3 are weakened. However, the previous conditions are enough to correct any error $ D(\mathbf{e}) \in \mathbb{F}_q^{n \times n} $, where $ \mathbf{e} \in \mathbb{F}_{q}^{n} $ and $ {\rm wt_H}(\mathbf{e}) \leq t $, 

Assume the received vector is $ R = D(\mathbf{c}) + D(\mathbf{e}) $, with $ \mathbf{c} \in \mathcal{C} $ and $ {\rm wt_H}(\mathbf{e}) \leq t $. Correcting the diagonal of $ R = D(\mathbf{c}) + D(\mathbf{e}) $ for the Hamming metric is the same as correcting the matrix $ R = D(\mathbf{c}) + D(\mathbf{e}) $ itself for the rank metric. We will next show that the algorithm in Theorem \ref{decoding type II} is exactly the same as the algorithm for ECPs in the Hamming metric.

Define $ I \subseteq \{ 1,2, \ldots, n \} $ as the Hamming support of $ \mathbf{e} \in \mathbb{F}_q^n $, that is, $ I  = {\rm HSupp}(\mathbf{e}) = \{ i \in \{ 1,2, \ldots, n \} \mid e_i \neq 0 \} $, and define
$$ \mathcal{K}_H(\mathbf{e}) = \{ \mathbf{a} \in \mathcal{A} \mid (\mathbf{b} \ast \mathbf{a}) \cdot \mathbf{e} = 0, \forall \mathbf{b} \in \mathcal{B} \}, \textrm{ and} $$
$$ \mathcal{A}(I) = \{ \mathbf{a} \in \mathcal{A} \mid {\rm HSupp}(\mathbf{a}) \subseteq I^c \}, $$
where $ I^c $ denotes the complementary of $ I $. It holds that $ {\rm Row}(D(\mathbf{e})) = \mathcal{L}_I \subseteq \mathbb{F}_q^n 
$, the space generated by the vectors $ \mathbf{e}_i $ in the canonical basis, for $ i \in I $. Therefore, by Remark \ref{products extension Hamming}, the properties of $ D $, Proposition \ref{connection vector products} and the fact that $ \mathcal{L}_I^\perp = \mathcal{L}_{I^c} $, it holds that
$$ \mathcal{K}(R) = \mathcal{K}(D(\mathbf{e})) = D(\mathcal{K}_H(\mathbf{e})) \quad \textrm{and} \quad (D(\mathcal{A}))(\mathcal{L}_I) = D(\mathcal{A}(I)). $$
Moreover, since $ \mathcal{A}(I) = \mathcal{K}_H(\mathbf{e}) $ by the properties of ECPs in the Hamming metric, we also have that
$$ \mathcal{K}(R) = D(\mathcal{K}_H(\mathbf{e})) = D(\mathcal{A}(I)) = (D(\mathcal{A}))(\mathcal{L}_I). $$
Hence, computing $ \mathcal{K}(R) $ implies computing $ (D(\mathcal{A}))(\mathcal{L}_I) $. Finally, since $ \mathcal{A}(I) \neq 0 $ by the properties of ECPs, we have that $ (D(\mathcal{A}))(\mathcal{L}_I) \neq 0 $. The rest of the algorithm goes in the same way as in Theorem \ref{decoding type II}. That is, the decoding algorithm in Theorem \ref{decoding type II} actually extends the decoding algorithm given by ECPs in the Hamming metric.

\subsection{Gabidulin codes}

Gabidulin codes, introduced in \cite{gabidulin}, are a well-known family of MRD $ \mathbb{F}_{q^m} $-linear codes in $ \mathbb{F}_{q^m}^n $, when $ n \leq m $. In \cite{new-construction}, a generalization of these codes is given, also formed by MRD codes.

Fix $ n \leq m $. They can be defined as follows. For each $ \mathbf{b} = (b_1,b_2, \ldots, b_n) \in \mathbb{F}_{q^m}^n $,  where $ b_1, b_2, \ldots, b_n $ are linearly independent over $ \mathbb{F}_q $, each $ k = 1,2, \ldots, n $ and each integer $ r $ such that $ r $ and $ m $ are coprime, we define the (generalized) Gabidulin code of dimension $ k $ in $ \mathbb{F}_{q^m}^n $ as
$$ {\rm Gab}_{k,m,n}(r, \mathbf{b}) = \{ (F(b_1), F(b_2), \ldots, F(b_n)) \mid F \in \mathcal{L}_{q,r,k} \mathbb{F}_{q^m}[x] \}, $$
where $ \mathcal{L}_{q,r,k} \mathbb{F}_{q^m}[x] $ denotes the $ \mathbb{F}_{q^m} $-linear space of $ q $-linearized polynomials of the form
$$ F(x) = a_0 x + a_1 x^{[r]} + a_2 x^{[2r]} + a_3 x^{[3r]} + \cdots + a_{k-1} x^{[(k-1)r]}, $$
for some $ a_0,a_1, \ldots, a_{k-1} \in \mathbb{F}_{q^m} $. Observe that classical Gabidulin codes as defined in  \cite{gabidulin} are obtained by setting $ r=1 $. Also observe that, for any invertible matrix $ P \in \mathbb{F}_q^{n \times n} $, it holds that
$$ {\rm Gab}_{k,m,n}(r, \mathbf{b}) P = {\rm Gab}_{k,m,n}(r, \mathbf{b}P), $$
and hence $ \mathbb{F}_{q^m} $-linearly rank-metric equivalent codes to Gabidulin codes are again Gabidulin codes.

The following lemma follows from Proposition \ref{evaluation properties}:

\begin{lemma} \label{gabidulin 1}
For every positive integers $ k,l $ with $ k + l -1 \leq n $, it holds that
$$ {\rm Gab}_{k,m,m}(r, \boldsymbol\alpha) \star {\rm Gab}_{l,m,n}(r, \mathbf{b}) = {\rm Gab}_{k+l-1,m,n}(r, \mathbf{b}). $$
In the case $ r=1 $, it holds that
$$ {\rm Gab}_{k,m,n}(1, \boldsymbol\alpha_n) \star {\rm Gab}_{l,m,n}(1, \mathbf{b}) = {\rm Gab}_{k+l-1,m,n}(1, \mathbf{b}). $$
\end{lemma}

On the other hand, for $ r=1 $ and the maps $ \varphi_n $, the following lemma follows from the definitions:

\begin{lemma} \label{gabidulin 2}
It holds that
$$ \varphi_n({\rm Gab}_{k,m,n}(1, \boldsymbol\alpha_n)) = {\rm Gab}_{k,m,m}(1, \boldsymbol\alpha). $$
\end{lemma}

With these two lemmas, we can prove that Gabidulin codes have $ t $-RECP of type I. Recall from \cite{new-construction} that
$$ {\rm Gab}_{k,m,n}(r, \mathbf{b})^\perp = {\rm Gab}_{n-k,m,n}(r, \mathbf{b}^\prime), $$
for some $ \mathbf{b}^\prime \in \mathbb{F}_{q^m}^n $ that can be computed from $ \mathbf{b} $.

\begin{theorem}
If $ t > 0 $, $ \mathcal{A} = {\rm Gab}_{t+1,m,n}(r, \mathbf{b}) $, $ \mathcal{B} = {\rm Gab}_{t,m,m}(r, \boldsymbol\alpha) $ and $ \mathcal{C} = {\rm Gab}_{2t,m,n}(r, \mathbf{b})^\perp $, then $ (\mathcal{A}, \mathcal{B}) $ is a $ t $-RECP of type I for $ \mathcal{C} $. In the case $ r=1 $, we may take $ \mathcal{B} = {\rm Gab}_{t,m,n}(r, \boldsymbol\alpha_n) $.
\end{theorem}
\begin{proof}
The first condition follows from Lemma \ref{gabidulin 1}. On the other hand, $ \dim_{\mathbb{F}_{q^m}}(\mathcal{A}) = t+1 $, so the second condition follows. The third condition is trivial, and for the case $ r=1 $ and $ \mathcal{B} = {\rm Gab}_{t,m,n}(1, \boldsymbol\alpha_n) $ it follows from Lemma \ref{gabidulin 2}. Finally, the fourth condition follows from the following computation:
$$ d_R(\mathcal{A}) + d_R(\mathcal{C}) = n-t + 2t+1 = n + t + 1. $$
\end{proof}

We see that $ d_R(\mathcal{A}) = n-t > n - 2t $. Hence, the pair $ (M_\alpha(\mathcal{A}), M_\alpha(\mathcal{B})) $, with notation as in Section \ref{connections}, can be used by Corollary \ref{stronger RECPs} to efficiently correct any error of rank at most $ t $ for every $ \mathbb{F}_q $-linear subcode of a (generalized) Gabidulin code. Such efficient decoding algorithms seem not to have been obtained yet.

\begin{corollary}
Let $ t, \mathcal{A}, \mathcal{B} $ and $ \mathcal{C} $ be as in the previous theorem. Then, for every $ \mathbb{F}_q $-linear subcode $ \mathcal{D} \subseteq \mathcal{C} $, the pair $ (M_\alpha(\mathcal{A}), M_\alpha(\mathcal{B})) $ is a $ t $-RECP of type II for $ M_{\alpha^\prime}(\mathcal{D}) $.
\end{corollary}
\begin{proof}
It follows from the previous theorem, Theorem \ref{RECP I admit RECP II} and Corollary \ref{stronger RECPs}.
\end{proof}

On the other hand, decoding algorithms for generalized Gabidulin codes with $ r \neq 1 $ seem to have been obtained only in \cite{new-construction}, also of cubic complexity.

\subsection{Skew cyclic codes}

Skew cyclic codes (or $ q^r $-cyclic codes) play the same role as cyclic codes in the theory of error-correcting codes for the rank metric. They were originally introduced in \cite{gabidulin} for $ r=1 $ and $ m=n $, and further generalized in \cite{qcyclic} for $ r=1 $ and any $ m $ and $ n $, and for any $ r $ in the work by Ulmer et al. \cite{skewcyclic1, skewcyclic3}. In this subsection we will only treat the case $ r=1 $.

Assume that $ n =sm $ is a multiple of $ m $. We will see in this subsection that, in that case, some $ \mathbb{F}_{q^m} $-linear $ q $-cyclic codes have rank error-locating pairs of type I, in analogy to the ideas in \cite{error-locating}. We say that an $ \mathbb{F}_{q^m} $-linear code $ \mathcal{C} \subseteq \mathbb{F}_{q^m}^n $ is $ q $-cyclic if the $ q $-shifted vector
$$ (c_{n-1}^{q}, c_0^{q}, c_1^{q}, \ldots, c_{n-2}^{q}) $$
lies in $ \mathcal{C} $, for every $ \mathbf{c} = (c_0, c_1, \ldots, c_{n-1}) \in \mathcal{C} $. As in \cite{rootskew}, we say that an $ \mathbb{F}_{q} $-linear subspace $ \mathcal{T} \subseteq \mathbb{F}_{q^{n}} $ is a $ q $-root space (over $ \mathbb{F}_{q^m} $) if it is the root space in $ \mathbb{F}_{q^{n}} $ of a linearized polynomial in $ \mathcal{L}_q \mathbb{F}_{q^m}[x] $.

By \cite[Theorem 3]{rootskew}, $ \mathbb{F}_{q^m} $-linear $ q $-cyclic codes are codes in $ \mathbb{F}_{q^m}^n $ with a parity check matrix over $ \mathbb{F}_{q^n} $ of the form
\begin{displaymath}
\mathcal{M}(\beta_1, \beta_2, \ldots, \beta_{n-k}) = \left(
\begin{array}{ccccc}
\beta_1 & \beta_1^{[1]} & \beta_1^{[2]} & \ldots & \beta_1^{[n-1]} \\
\beta_2 & \beta_2^{[1]} & \beta_2^{[2]} & \ldots & \beta_2^{[n-1]} \\
\vdots & \vdots & \vdots & \ddots & \vdots \\
\beta_{n-k} & \beta_{n-k}^{[1]} & \beta_{n-k}^{[2]} & \ldots & \beta_{n-k}^{[n-1]} \\
\end{array} \right),
\end{displaymath}
where $ \beta_1, \beta_2, \ldots, \beta_{n-k} $ is a basis of $ \mathcal{T} $ over $ \mathbb{F}_{q} $, for some $ q $-root space $ \mathcal{T} $. Moreover by \cite[Corollary 2]{rootskew}, $ \mathbb{F}_{q^m} $-linear $ q $-cyclic codes are in bijection with $ q $-root spaces over $ \mathbb{F}_{q^m} $. 

The next bound, which is given in \cite[Corollary 4]{rootskew}, is an extension of the rank-metric version of the BCH bound (by setting $ w=0 $ and $ c=1 $) found in \cite[Proposition 1]{skewcyclic3}:

\begin{lemma}[\textbf{Rank-HT bound}] \label{rank bound}
Let $ b $, $ c $, $ \delta $ and $ w $ be positive integers with $ \delta + w \leq m $ and $ d = {\rm gcd}(c,n) < \delta $, and $ \alpha \in \mathbb{F}_{q^{n}} $ be such that the set $ \mathcal{A} = \{ \alpha^{[b + i + jc]} \mid 0 \leq i \leq \delta -2, 0 \leq j \leq w \} $ is a linearly independent set of vectors. 

If $ \mathcal{C} $ is the $ \mathbb{F}_{q^m} $-linear $ q $-cyclic code corresponding to the $ q $-root space $ \mathcal{T} $ and $ \mathcal{A} \subseteq \mathcal{T} $, then $ d_R(\mathcal{C}) \geq \delta + w $.
\end{lemma}

To use it, we need to deal with normal bases. First, it is well-known \cite{lidl} that the orthogonal (or dual) basis of a normal basis $ \alpha, \alpha^{[1]}, \ldots, \alpha^{[n-1]} \in \mathbb{F}_{q^n} $ over $ \mathbb{F}_q $ is again a normal basis $ \beta, \beta^{[1]}, \ldots, \beta^{[n-1]} \in \mathbb{F}_{q^n} $. Define $ \boldsymbol\alpha = (\alpha, \alpha^{[1]}, \ldots, \alpha^{[n-1]}) $ and $ \boldsymbol\beta = (\beta, \beta^{[1]}, \ldots, \beta^{[n-1]}) $. Then it holds that
$$ \boldsymbol\alpha^{[i]} \cdot \boldsymbol\beta^{[j]} = {\rm Tr}(\alpha^{[i]} \beta^{[j]}) = \delta_{i,j} $$
by definition. On the other hand, for a subset $ I \subseteq \{ 1,2, \ldots, n \} $, define the matrix
$$ \mathcal{M}_{\boldsymbol\alpha}(I) = \mathcal{M}(\alpha^{[i]} \mid i \in I), $$
and similarly for $ \boldsymbol\beta $.

Define the $ \mathbb{F}_{q^m} $-linear codes $ \mathcal{A}, \mathcal{B} \subseteq \mathbb{F}_{q^m}^n $ as the subfield subcodes of the codes in $ \mathbb{F}_{q^n}^n $ with generator matrices $ \mathcal{M}_{\boldsymbol\alpha}(I) $ and $ \mathcal{M}_{\boldsymbol\alpha}(J) $, for some subsets $ I, J \subseteq \{ 1,2, \ldots, n \} $, respectively. 

In order to obtain $ q $-cyclic codes, we will assume that the space generated by $ \{ \alpha^{[i]} \mid i \in I \} $ is a $ q $-root space, and similarly for $ J $. Due to the cyclotomic space description of $ q $-root spaces in \cite[Proposition 2]{rootskew}, this holds if the following condition holds: if $ i \in I $, then $ i + m \in I $ (modulo $ n $), and similarly for $ J $.

Define the $ \mathbb{F}_{q^m} $-linear $ q $-cyclic code $ \mathcal{C} \subseteq \mathbb{F}_{q^m}^n $ with parity check matrix $ \mathcal{M}_{\boldsymbol\alpha}(I + J) $. Observe that $ I+J $ also gives a $ q $-root space by the previous paragraph. We have the following lemmas:

\begin{lemma}
$ \mathcal{A} $ and $ \mathcal{B} $ are the $ q $-cyclic codes with parity check matrices $ \mathcal{M}_{\boldsymbol\beta}(I^c) $ and $ \mathcal{M}_{\boldsymbol\beta}(J^c) $ over $ \mathbb{F}_{q^n} $, respectively.
\end{lemma}
\begin{proof}
We prove it for $ \mathcal{A} $. Define $ \widetilde{\mathcal{A}} $ as the $ \mathbb{F}_{q^n} $-linear code in $ \mathbb{F}_{q^n}^n $ with generator matrix $ \mathcal{M}_{\boldsymbol\alpha}(I) $. It is enough to prove that $ \mathcal{M}_{\boldsymbol\beta}(I^c) $ is a parity check matrix for $ \widetilde{\mathcal{A}} $. 

However, since $ \boldsymbol\alpha^{[i]} \cdot \boldsymbol\beta^{[j]} = 0 $, for every $ i \in I $ and $ j \notin I $, it holds that $ \mathcal{M}_{\boldsymbol\alpha}(I) \mathcal{M}_{\boldsymbol\beta}(I^c)^T = 0 $. On the other hand, these two matrices are full rank and the number of rows in $ \mathcal{M}_{\boldsymbol\alpha}(I) $ together with the number of rows in $ \mathcal{M}_{\boldsymbol\beta}(I^c) $ is $ \# I + \# (I^c) = n $, and the result follows.
\end{proof}

\begin{lemma}
It holds that $ \mathcal{B} \star \mathcal{A} \subseteq \mathcal{C}^\perp $.
\end{lemma}
\begin{proof}
By Proposition \ref{evaluation properties}, item 2, we see that $ \mathcal{B} \star \mathcal{A} $ is contained in the $ \mathbb{F}_{q^n} $-linear code with generator matrix $ \mathcal{M}_{\boldsymbol\alpha}(I + J) $. Denote such code by $ \mathcal{D} $, that is, $ \mathcal{B} \star \mathcal{A} \subseteq \mathcal{D} $ and $ \mathcal{D} \subseteq \mathbb{F}_{q^n}^n $.

By definition, $ \mathcal{C} = \mathcal{D}^\perp \cap \mathbb{F}_{q^m}^n $, and by \cite[Corollary 3]{rootskew}, $ \mathcal{D} $ is Galois closed over $ \mathbb{F}_{q^m} $, which means that $ \mathcal{D}^\perp \cap \mathbb{F}_{q^m}^n = (\mathcal{D} \cap \mathbb{F}_{q^m}^n)^\perp $ by \cite[Proposition 2]{similarities} and Delsarte's theorem \cite[Theorem 2]{delsarte}. Hence 
$$ \mathcal{B} \star \mathcal{A} \subseteq \mathcal{D} \cap \mathbb{F}_{q^m}^n = (\mathcal{D}^\perp \cap \mathbb{F}_{q^m}^n)^\perp = \mathcal{C}^\perp. $$
\end{proof}

We may now prove that $ (\mathcal{A}, \mathcal{B}) $ is a $ t $-rank error-locating pair of type I for $ \mathcal{C} $, and with some stronger hypotheses, it is also a $ t $-rank error-correcting pair for $ \mathcal{C} $.

\begin{theorem}
Fix a positive integer $ t $ and assume that $ \#I > t $ and $ J $ contains $ \delta - 1 $ consecutive elements, for some $ \delta > t $. Then $ (\mathcal{A}, \mathcal{B}) $ is a $ t $-rank error-locating pair for $ \mathcal{C} $ (see Remark \ref{error-locating I}). If moreover, $ d_R(\mathcal{A}) + d_R(\mathcal{C}) > n $, then $ (\mathcal{A}, \mathcal{B}) $ is a $ t $-rank error-correcting pair of type I for $ \mathcal{C} $.
\end{theorem}
\begin{proof}
From the previous lemma, we have that $ \mathcal{B} \star \mathcal{A} \subseteq \mathcal{C}^\perp $. On the other hand, $ \mathcal{A} $ satisfies that $ \dim(\mathcal{A}) = \#I > t $, and $ \mathcal{B} $ satisfies that $ d_R(\mathcal{B}^\perp) \geq \delta > t $ by Lemma \ref{rank bound}. Then $ (\mathcal{A}, \mathcal{B}) $ is a $ t $-rank error-locating pair of type I for $ \mathcal{C} $.
\end{proof}

Observe that we may obtain the bound $ d_R(\mathcal{C}) \geq \delta + w $ by Proposition \ref{first bound} assuming that $ I $ contains the elements $ ic $, for $ 0 \leq i \leq w $. This means that, for $ q $-cyclic codes constructed with a normal basis, the rank-HT bound found in \cite[Corollary 4]{rootskew} is implied by Proposition \ref{first bound}, as in the classical case. Further cases are left open.

\section*{Acknowledgement}

Rank error-correcting pairs of type II have been obtained in the case $ m = n $ independently by Alain Couvreur \cite{couvreur}. The authors wish to thank him for this communication and for useful discussions and comments.

The authors also gratefully acknowledge the support from The Danish Council for Independent Research (Grant No. DFF-4002-00367). This paper was started during the visit of the second author to Aalborg University, Denmark, which was supported by the previous grant.


\begin{thebibliography}{10}

\bibitem{berger}
T.P. Berger.
\newblock Isometries for rank distance and permutation group of {G}abidulin
  codes.
\newblock {\em IEEE Transactions Information Theory}, 49(11):3016--3019, 2003.

\bibitem{skewcyclic1}
D.~Boucher, W.~Geiselmann, and F.~Ulmer.
\newblock Skew-cyclic codes.
\newblock {\em Applicable Algebra in Engineering, Communication and Computing},
  18(4):379--389, 2007.

\bibitem{skewcyclic3}
L.~Chaussade, P.~Loidreau, and F.~Ulmer.
\newblock Skew codes of prescribed distance or rank.
\newblock {\em Designs, Codes and Cryptography}, 50(3):267--284, 2009.

\bibitem{couvreur}
A.~Couvreur.
\newblock About error correcting pairs. {P}ersonal communication.
\newblock September 15, 2015.

\bibitem{delsartebilinear}
P.~Delsarte.
\newblock Bilinear forms over a finite field, with applications to coding
  theory.
\newblock {\em Journal of Combinatorial Theory, Series A}, 25(3):226 -- 241,
  1978.

\bibitem{delsarte}
P.~Delsarte.
\newblock On subfield subcodes of modified reed-solomon codes (corresp.).
\newblock {\em IEEE Transactions Information Theory}, 21(5):575--576, September
  2006.

\bibitem{error-locating}
I.M. Duursma and R.~K{\"o}tter.
\newblock Error-locating pairs for cyclic codes.
\newblock {\em IEEE Transactions on Information Theory}, 40(4):1108--1121,
  1994.

\bibitem{pellikaan-duursma}
I.M. Duursma and R.~Pellikaan.
\newblock A symmetric {R}oos bound for linear codes.
\newblock {\em Journal of Combinatorial Theory, Series A}, 113(8):1677--1688,
  2006.

\bibitem{gabidulin}
E.M. Gabidulin.
\newblock Theory of codes with maximum rank distance.
\newblock {\em Problems Information Transmission}, 21, 1985.

\bibitem{qcyclic}
E.M. Gabidulin.
\newblock Rank q-cyclic and pseudo-q-cyclic codes.
\newblock In {\em IEEE International Symposium on Information Theory, 2009.
  ISIT 2009}, pages 2799--2802, 2009.

\bibitem{hartmann}
C.R.P. Hartmann and K.K. Tzeng.
\newblock Generalizations of the {BCH} bound.
\newblock {\em Information and Control}, 20(5):489 -- 498, 1972.

\bibitem{jurrius-pellikaan}
R.P.M.J. Jurrius and R.~Pellikaan.
\newblock The extended and generalized rank weight enumerator.
\newblock ACA 2014, Applications of Computer Algebra, 9 July 2014, Fordham
  University, New York, Computer Algebra in Coding Theory and Cryptography,
  2014.

\bibitem{slides}
R.P.M.J. Jurrius and R.~Pellikaan.
\newblock On defining generalized rank weights.
\newblock {\em arXiv:1506.02865}, 2015.

\bibitem{unified}
R.~K{\"o}tter.
\newblock A unified description of an error locating procedure for linear
  codes.
\newblock {\em Proceedings of Algebraic and Combinatorial Coding Theory}, pages
  113 -- 117, 1992.
\newblock Voneshta Voda.

\bibitem{new-construction}
A.~Kshevetskiy and E.M. Gabidulin.
\newblock The new construction of rank codes.
\newblock In {\em IEEE International Symposium on Information Theory, 2005.
  ISIT 2005}, pages 2105--2108, 2005.

\bibitem{lidl}
R.~Lidl and H.~Niederreiter.
\newblock {\em Finite Fields.}, volume~20.
\newblock Encyclopedia of Mathematics and its Applications. Addison-Wesley,
  Amsterdam, 1956.

\bibitem{loidreau-gabidulin}
Pierre Loidreau.
\newblock A {W}elch–{B}erlekamp like algorithm for decoding {G}abidulin
  codes.
\newblock In Øyvind Ytrehus, editor, {\em Coding and Cryptography}, volume
  3969 of {\em Lecture Notes in Computer Science}, pages 36--45. Springer
  Berlin Heidelberg, 2006.

\bibitem{rootskew}
U.~Mart{\'i}nez-Pe{\~n}as.
\newblock On the roots and minimum rank distance of skew cyclic codes.
\newblock {\em arXiv:1511.09329}, 2015.

\bibitem{similarities}
U.~Mart{\'i}nez-Pe{\~n}as.
\newblock On the similarities between generalized rank and {H}amming weights
  and their applications to network coding.
\newblock {\em arXiv:1506.04036}, 2015.

\bibitem{on-ECP}
R.~Pellikaan.
\newblock On decoding linear codes by error correcting pairs.
\newblock {\em Preprint. Eindhoven University of Technology}, 1988.

\bibitem{pellikaan-error-location}
R.~Pellikaan.
\newblock On decoding by error location and dependent sets of error positions.
\newblock {\em Discrete Mathematics}, 106:369--381, 1992.

\bibitem{on-the-existence}
R.~Pellikaan.
\newblock On the existence of error-correcting pairs.
\newblock {\em Journal of Statistical Planning and Inference}, 51(2):229 --
  242, 1996.
\newblock Shanghai Conference Issue on Designs, Codes, and Finite Geometries,
  Part I.

\bibitem{ravagnani}
A.~Ravagnani.
\newblock Rank-metric codes and their duality theory.
\newblock {\em Designs, Codes and Cryptography}, pages 1--20, 2015.

\bibitem{roos1}
C.~Roos.
\newblock A generalization of the {BCH} bound for cyclic codes, including the
  {H}artmann-{T}zeng bound.
\newblock {\em Journal Combinatorial Theory (Series A)}, 33:229--232, 1982.

\bibitem{roos2}
C.~Roos.
\newblock A new lower bound for the minimum distance of a cyclic code.
\newblock {\em IEEE Transactions Information Theory}, IT-29:330--332, 1982.

\bibitem{on-metrics}
D.~Silva and F.R. Kschischang.
\newblock On metrics for error correction in network coding.
\newblock {\em IEEE Transactions Information Theory}, 55(12):5479--5490, 2009.

\end{thebibliography}

\def\cprime{$'$}

\end{document}